\documentclass[format=acmsmall, review=false]{acmart}
\usepackage{acm-ec-23}
\usepackage{booktabs} 
\usepackage[ruled]{algorithm2e} 

\SetAlFnt{\small}
\SetAlCapFnt{\small}
\SetAlCapNameFnt{\small}
\SetAlCapHSkip{0pt}
\IncMargin{-\parindent}
\usepackage{caption}
\usepackage{subcaption}

\usepackage{comment}

\usepackage{graphicx}
\usepackage{xcolor}
\definecolor{darkgreen}{HTML}{004b23}
\definecolor{green}{HTML}{38b000}
\definecolor{lightgreen}{HTML}{aacc00}
\colorlet{red}{red}
\colorlet{navy}{black!40!blue}
\definecolor{blue}{HTML}{1e90ff}
\definecolor{lightblue}{HTML}{87CEEB}
\colorlet{orange}{orange}
\colorlet{yellow}{red!20!yellow}
\definecolor{lavendar}{HTML}{ba55d3}
\newcommand{\revcircle}[3]{\raisebox{0.5pt}{\tikz{\node[draw=#3,thick,scale=#1,circle,fill=#2](){};}}}
\newcommand*{\isotri}[1]{%
  \scalebox{1.118034}[1]{$#1$}%
}

\usepackage{microtype}
\usepackage{nicefrac}
\usepackage{amsthm}
\usepackage{thm-restate}
\usepackage{amsmath}
\usepackage{mathtools}
\usepackage{thmtools} 
\newtheorem{theorem}{Theorem}[section]%
\newtheorem{lemma}[theorem]{Lemma}%
\newtheorem{corollary}[theorem]{Corollary}%

\newtheorem{definition}[theorem]{Definition}

\newtheorem{proposition}[theorem]{Proposition}

\newcommand{\AutoAdjust}[3]{\mathchoice{ \left #1 #2  \right #3}{#1 #2 #3}{#1 #2 #3}{#1 #2 #3} }
\newcommand{\Xcomment}[1]{{}}

\newcommand\diff{\mathop{}\!\mathrm{d}}

\newcommand{\InBrackets}[1]{\AutoAdjust{[}{#1}{]}}

\newcommand{\Ex}[2][]{\operatorname{\mathbf E}_{#1}\InBrackets{#2}}
\newcommand{\Prx}[2][]{\operatorname{\mathbf{Pr}}_{#1}\InBrackets{#2}}
\def\prob{\Prx}
\def\expect{\Ex}

\newcommand{\noaccents}[1]{#1}

\newcommand{\newagentvar}[3][\noaccents]{%
\expandafter\newcommand\expandafter{\csname #2\endcsname}{#1{#3}}%
\expandafter\newcommand\expandafter{\csname #2s\endcsname}{#1{\boldsymbol{#3}}}%
\expandafter\newcommand\expandafter{\csname #2smi\endcsname}[1][i]{#1{\boldsymbol{#3}}_{-##1}}%
\expandafter\newcommand\expandafter{\csname #2i\endcsname}[1][i]{#1{#3}_{##1}}%
\expandafter\newcommand\expandafter{\csname #2k\endcsname}[1][i]{#1{#3}_{##1}}%
\expandafter\newcommand\expandafter{\csname #2ith\endcsname}[1]{#1{#3}^{(##1)}}%
}

\DeclareMathOperator{\REV}{\mathcal{R}}
\DeclareMathOperator{\rev}{rev}
\DeclareMathOperator{\myer}{myer}

\DeclareMathOperator{\UTIL}{\mathcal{U}}

\DeclareMathOperator{\supp}{supp}

\DeclareMathOperator*{\argmax}{arg\,max}

\newagentvar{typespace}{T}
\newagentvar{type}{\gamma}
\newagentvar{typeDist}{F_{\gamma}}
\newagentvar{typeDense}{f_{\gamma}}
\newagentvar{repType}{\hat t}
\newagentvar{dist}{F}

\newagentvar{dense}{f}
\newagentvar{margDist}{F_{m}}

\newagentvar{margDense}{f_{m}}
\newagentvar{incDist}{W}
\newagentvar{valueFunc}{v}
\newagentvar{pnosale}{p_{\mathrm{end}}}

\newagentvar{valuePay}{\psi}
\newagentvar{utilityFunc}{u}
\newagentvar{utilityPay}{\mu}
\newagentvar{costFunc}{c}
\newagentvar{initVal}{q}
\newagentvar{incValue}{w}
\newagentvar{projValue}{V}
\newagentvar{skp}{S}

\newagentvar{startslope}{m_{1}}
\newagentvar{slopeX}{m_{x}}
\newagentvar{slope}{m}

\newagentvar{skippay}{p}
\newagentvar{skipret}{p_{r}}
\newagentvar{skipcond}{p_{c}}
\newagentvar{partpay}{q}
\newagentvar{payment}{p} 
\newagentvar{utility}{u}
\newagentvar{virt}{\psi}

\newcommand{\task}{k}

\newcommand{\retention}{r}
\newcommand{\margValue}{v_m}

\newcommand{\retDist}{F_r}
\newcommand{\retDens}{f_{r}}

\newcommand{\gDiscount}{\beta}
\newcommand{\discountDense}{f_{\gamma}}

\makeatletter
\newcommand{\taylor}[1]{\textcolor{blue}{[Taylor says\@ifnotempty{#1}{: #1}]}}
\newcommand{\narun}[1]{\textcolor{green}{[Narun says\@ifnotempty{#1}{: #1}]}}
\newcommand{\hufu}[1]{\textcolor{lavendar}{[Hu says\@ifnotempty{#1}{: #1}]}}
\newcommand{\klb}[1]{\textcolor{red}{[Kevin says\@ifnotempty{#1}{: #1}]}}

\makeatother
\usepackage{xargs}
\usepackage{tikz}
\usepackage{cleveref}
\newcommand\numberthis{\addtocounter{equation}{1}\tag{\theequation}}
\setcitestyle{authoryear}
\title{Pay to (Not) Play: Monetizing Impatience in Mobile Games}

\author{Taylor Lundy}
\email{tlundy@cs.ubc.ca}
\affiliation{%
\institution{University of British Columbia}
\city{Vancouver}
\country{Canada}
}
\author{Narun Raman}
\email{narunram@cs.ubc.ca}
\affiliation{%
\institution{University of British Columbia}
\city{Vancouver}
\country{Canada}
}
\author{Hu Fu}
\email{fuhu@mail.shufe.edu.cn}
\affiliation{%
\institution{Shanghai University of Finance and Economics}
\city{Shanghai}
\country{China}}
\affiliation{%
\institution{Key Laboratory of Interdisciplinary Research of Computation and Economics} 
\department{Ministry of Education}
\city{Shanghai}
\country{China}
}
\author{Kevin Leyton-Brown}
\email{kevinlb@cs.ubc.ca}
\affiliation{%
\institution{University of British Columbia}
\city{Vancouver}
\country{Canada}
}

\begin{abstract}
    Mobile gaming is a rapidly growing and incredibly profitable sector; having grown seven-fold over the past 10 years, it now grosses over \$100 billion annually.
    This growth was due in large part to a shift in monetization strategies: rather than charging players an upfront cost (``pay-to-play''), games often request optional microtransactions throughout gameplay (``free-to-play''). 
    We focus on a common scenario in which games include wait times---gating either items or game progression---that players can pay to skip. 
    Game designers typically say that they optimize for player happiness rather than revenue; however, prices for skips are typically set at levels that few players are willing to pay, leading to low purchase rates. 
    Under a traditional analysis, it would seem that game designers fail at their stated goal if few players buy what they are selling. 
    We argue that an alternate model can better explain this dynamic: players value tasks more highly as they are perceived to be more difficult. 
    While skips can increase players' utilities by providing instant gratification, pricing skips too cheaply can \emph{lower} players' utilities by decreasing the perceived amount of work needed to complete a task.
    We show that high revenue, high player utility, and low purchase rates can all coexist under this model, particularly under a realistic distribution of players having few buyers but a few big-spending ``whales.''
    We also investigate how a game designer should optimize prices under our model.
\end{abstract}

\begin{document}

\maketitle

\section{Introduction}

The gaming industry is larger than Hollywood and the music industry combined; currently valued at \$245.10 billion USD, it is expected to reach \$376.08 billion by 2028 \cite{mordor2023}. 
While console and PC games have historically dominated the gaming category, mobile games have achieved seven-fold revenue growth since 2012 \cite{newzoo2021} and accounted for a majority of gaming industry revenues in 2023 \cite{mordor2023}. 
Reflecting this shifting marketplace, Microsoft recently announced plans to acquire Activision Blizzard for \$68.7 billion USD \cite{microsoft_2022}.
Just two months after release, \emph{Diablo Immortal}, Blizzard's franchise mobile title, surpassed \$100 million USD; overall, mobile games account for 40\% (\$1.08 billion USD) of Activision's revenue \cite{partis_2022, activision_2021}.
 
A key factor in the rise of mobile games is a shift in pricing models. 
Over 90\% of mobile games are free to download (``free-to-play'') and monetized by offering optional ``microtransactional'' payments throughout game-play \cite{misiuk_2019, seufert_2017, Tafradzhiyski_2023}.
While free-to-play is not a novel strategy (e.g., many platforms like Spotify offer free tiers), mobile games monetize a qualitatively different set of behaviours.
For example, mobile games often allow users to spend money on virtual cosmetic items, extra lives that enable longer play sessions, or items that make the game easier.
We are interested in understanding what player behavior causes these strategies and, in this paper, we focus on skips.

Games that monetize using skips have a gameplay loop requiring ``grinding'': completing tedious tasks, such as waiting for a timer to tick down. 
Players are offered the ability to skip the grind for a price.
This style of monetization is so successful that an entire genre of games consisting solely of sequential timers and skips generated over \$3 billion worldwide in 2021 \cite{gdcriseofidle2013, moloco2021}.

While the success of skip-based monetization may be surprising, it is odder still how skips are priced.
It is common to see talks at game design conferences discussing the perils of short-sighted monetization strategies \cite{Greer_2013, Nussbaum_Telfer_2015, krasilnikov_2020, Reilley_2023}, or the importance of ensuring that players feel that purchases were ``worth it'' \cite{levy_2021, engblom_2017, Reilley_2023}.
Successful game designers understand this well; when talking about fixing stagnating revenues in Clash of Clans, the chief game lead at the time, Eino Joas \citeyearpar{joas_2020} said, ``retention always trumps monetization''. 
Despite this, virtual goods such as skips are typically priced so high in practice that very few people are willing to buy them.
On average, only 3.5\% of players spend any money in-game and the majority of converted players, or `minnows,' spend only \$1 to \$5 a month and contribute to less than 15\% of the revenue \cite{brightman_2016, liftoff_2020}.
The majority of revenue comes from `whales,' who spend more than \$25 per month on average and account for less than 15\% of buyers \cite{shi2015minnow}.
Under traditional economic models, high average utilities and prices few buyers can afford are incompatible. 
It might, therefore, seem either that designers are failing to make most players happy or that they truly care more about revenue than player utility.

We show that existing game designs can indeed yield high player utility---along with high revenue---under an alternative model of player value.
In this model, each individual player values skips because they dislike waiting. 
However, skips can decrease \emph{all} players' values for virtual goods the more they are used.
We believe that players value signaling that they achieved a high rank by grinding for it and cheap skips degrade this value by making progress purchasable.
Modeling how the purchasing of skips degrades value requires modeling both how knowledgeable players are about other players' purchasing habits and how they incorporate this knowledge into their valuations. 
We make these modeling decisions based on the design dimensions of the game: the more immersive or social a game is the more sensitive players' values for playing are to different signals of grinding. 
We use this model to explain existing pricing strategies, but also show that, under some simplifying assumptions, game designers can leverage tools from traditional mechanism design to optimize prices.

We now describe what follows in the remainder of this paper.
\Cref{sec:relatedwork} considers related work studying utility models grounded in difficulty, cost, or scarcity; offering evidence about ways in which mobile game players vary; and studying the monetization of mobile games.
\Cref{sec:model} formally introduces our model of player utility where players have value functions that increase with the rarity of skip usage.
We model three types of mobile game players, differing in their responsiveness to other players' behaviors: fully-sensitive, price-sensitive, and insensitive.
With our model in hand we offer two main results.
First, in \Cref{sec:util}, we demonstrate that, when players' value functions are price-sensitive, our model can explain the counter-intuitive pricing observed in practice, giving conditions under which high prices yield both good utility and good revenue. 
Furthermore, we show that these conditions are satisfied under player distributions with few buyers but a sufficient number of ``whales.''
We illustrate these theoretical results in several examples and show that they also hold empirically in the more complicated fully-sensitive setting.
Second, in \Cref{sec:rev}, we prove that when players' value functions are insensitive, a simple pricing scheme approximates optimal revenue. 
We also show that this same pricing scheme approximates optimal revenue even when value functions are price-sensitive.
We then test the simple pricing scheme with simulations, showing that in a variety of different settings, this scheme is competitive with other, natural pricing schemes. 
Finally, we conclude our work in \Cref{sec:future}.

\section{Related Work}\label{sec:relatedwork}

A long line of literature---both in game theory and beyond---explores the idea that an object's value to an agent can be impacted by the difficulty or cost of obtaining it. 
For example, this idea is implicit in the idea of conspicuous consumption \cite{veblen1955, giffen1991}, which posits that agents value goods in part because they can be used to signal wealth.
In such models, higher prices and scarcity are primary drivers of value. 
Conspicuous consumption in mobile gaming is a well-documented phenomenon \cite{martin2008consuming, lehdonvirta2009virtual, cai2019purchases}.
Directly related to our work, \citet{geng2019conspicuous} investigated the freemium business model, finding that conspicuous consumption is a main driver of premium content pricing.
In other cases, players can be motivated to signal their skill or commitment to a game rather than their wealth \cite{carter2016cheating}.
Further work has shown how social factors can exacerbate conspicuous consumption \cite{cass2004status}.
Realizing this, mobile games often incorporate social networks to reinforce the conspicuous value of the virtual goods they offer \cite{fields2011social, choi2004people, grieve2013face}. 
\citet{goetz2022peer} use data from online games to show that being matched with other players who have acquired virtual items makes a given player more likely to purchase them.

Of course, a player's perception of value also depends on intrinsic factors, such as their degree of impatience.
In this vein, \citet{hanner2015purchasing} showed how player purchasing behaviour can be predicted based on past purchases. 
Additionally, these intrinsic factors can lead to very different player behaviour. 
Game designers often report that in real-world player populations, the majority of revenue comes from a select group (around 2-3\%) of ``whale'' players \cite{gdcgamedata, shi2015minnow, liftoff_2020}.

A key theme of this paper is that optimal pricing in mobile games must carefully trade off the value players assign to rewards that are difficult to obtain and the impatience that they experience as a result.
The relationship between these two factors has been studied in various other contexts. 
\citet{hsiao2015intent,evans2016economics,keogh2018waiting} study how these factors affect the motivations behind purchases in mobile games. 
Recent work by \citet{sepulveda2019exploring} looks at how optimizing game difficulty can impact player satisfaction.

Player happiness is a first-order consideration in the monetization of mobile games: games that rely on microtransactions need players to keep coming back to play more and pay more. 
The link between retention and revenue has been well studied in both academia and the game industry, with many approaches for increasing retention being empirically validated and widely deployed \cite{nojima2007MMO,gdclongterm2013}.
For example, \citet{xue2017dynamic,huang2019level} study how to adjust the difficulty of a game or opponents to keep players from leaving. 

Payments for skipping past obstacles yield considerable revenue in mobile games \cite{gdcriseofidle2013}.
This style of difficulty-reducing monetization has started to gain attention in the optimization literature. 
\citet{sheng2020incentivized} study the revenue and retention implications of offering items that make the game easier, focusing on games that provide these items in exchange for watching an ad. 
Further afield, \citet{jiao2021selling} study the impact of selling skill-boosting items in multiplayer, competitive games. 

Other recent papers study additional challenges related to game monetization that overlap less with our own work. 
Examples include optimizing pricing for loot boxes, which are lotteries over virtual items \cite{chen2019loot}, and characterizing tradeoffs between perpetual and subscription-based licensing \cite{dierks2020revenue}.

\section{Our Model}\label{sec:model}
The most straightforward rationale for players valuing game playing is that the game mechanic itself is enjoyable: it is fun to blast meteors out of the sky; to explore an open world; to solve puzzles; etc. 
Many mobile games have mechanics that defy such an explanation. 
Instead, in order to complete a level or unlock an item, players must simply wait for a timer to tick down.
It is hard to believe that there is anything fun about this kind of game mechanic (e.g., most people dislike waiting for a bus or a checkout clerk). 
Furthermore, players would not value the opportunity to skip timers if waiting were what made the game enjoyable. 

Our model rests on a sharply different assumption: that players value achieving tasks because of the perceived difficulty of completing them. 
This means that as achievements become scarcer in the player population, their value increases.
Farm\-Ville crops that take days to grow have high value because they are hard to obtain; a Town Hall Level 12 in Clash of Clans is desirable because obtaining one requires $2$ years and $2$ days of waiting through timers \cite{pixel_crux}. 
This is similar to existing models of conspicuous consumption except that the cost is denominated in time rather than money.

While there is value in difficult gameplay, there is also frustration. 
We assume that a player's demand for skips comes from a time discount in the reward for completing the task, which we model as their type $\type \in [0,1]$. We allow for players to differ in their impatience: they draw their types privately and i.i.d.\@ from some prior continuous distribution $\typeDist$. 
Furthermore, value in difficult gameplay can be shaped by the way other players in the player base accomplish their own tasks: players trying to signal their commitment to grinding care about how they compare to others.
Skip prices can therefore affect the conspicuous value of a task, as lower skip prices decrease the fraction of players who accomplished a task by waiting. 
There is typically no way to tell if another player completed a task by waiting or by skipping.
Thus, the perceived value of a task falls with the cost of skipping it, even if the number of players who accomplish it is held constant.

How players' values for signalling grinding are impacted by skip prices depends heavily on how  games are designed. 
We divide mobile games into three broad categories---mid-core, casual, and hypercasual---and argue that each gives rise to a different level of sensitivity to other players' behavior.
First, mid-core games (e.g., Clash of Clans, Game of War) incorporate in-game communities as a way to foster conspicuous value and immersion. 
For instance, a crucial part of game-playing in Clash of Clans is joining `clans' with other players.
Players can see the progress of other members in their clan and easily deduce if a clan member's rank was gained by purchasing or waiting.
We say that players in these games are \emph{fully-sensitive} and assume they know the full type distribution.
We assume that fully-sensitive players determine their value based on both the price of a skip as well as the type distribution of other players.

\begin{definition}\label{def:projval_full}
A fully-sensitive player with type $\type$ drawn from type distribution $\typeDist$, at the beginning of a task has projected value $\type\valueFunc(\typeDist, \skippay)$ for finishing the task without skipping, when the skip price is $\skippay$. 
\end{definition}

We insist that these value functions maintain some nice properties with respect to price.
First, the limit case of zero-cost skips is nonsensical (they would always be chosen), so we assume that $\valueFunc(\typeDist, 0) = 0$.
We also assume that value is monotonically increasing and concave in~$\skippay$. 
These two conditions mean that as the price increases, the probability of sale for skips weakly decreases and the value of the accomplishment weakly increases.
To avoid the uninteresting cases where players obtain non-positive utility at all skip prices we assume that the $\valueFunc'(\typeDist, 0)$ is greater than 1.
Finally, we assume that there exists a point $\pnosale > 0$ where $\valueFunc(\typeDist, \pnosale) = \pnosale$ at which point the value function becomes a constant $\pnosale$.
Otherwise, the designer can give players arbitrarily high utility by setting prices arbitrarily high. 
Figure \ref{fig:valuefunc} illustrates two value functions that satisfy the properties above, with the shaded region indicating values for $\valueFunc(\typeDist, \skippay)$ that yield negative utility at price~$\skippay$.

\begin{figure}[t]
    \centering
    \includegraphics[trim={0 0cm 0 1cm},clip, width=0.4\textwidth]{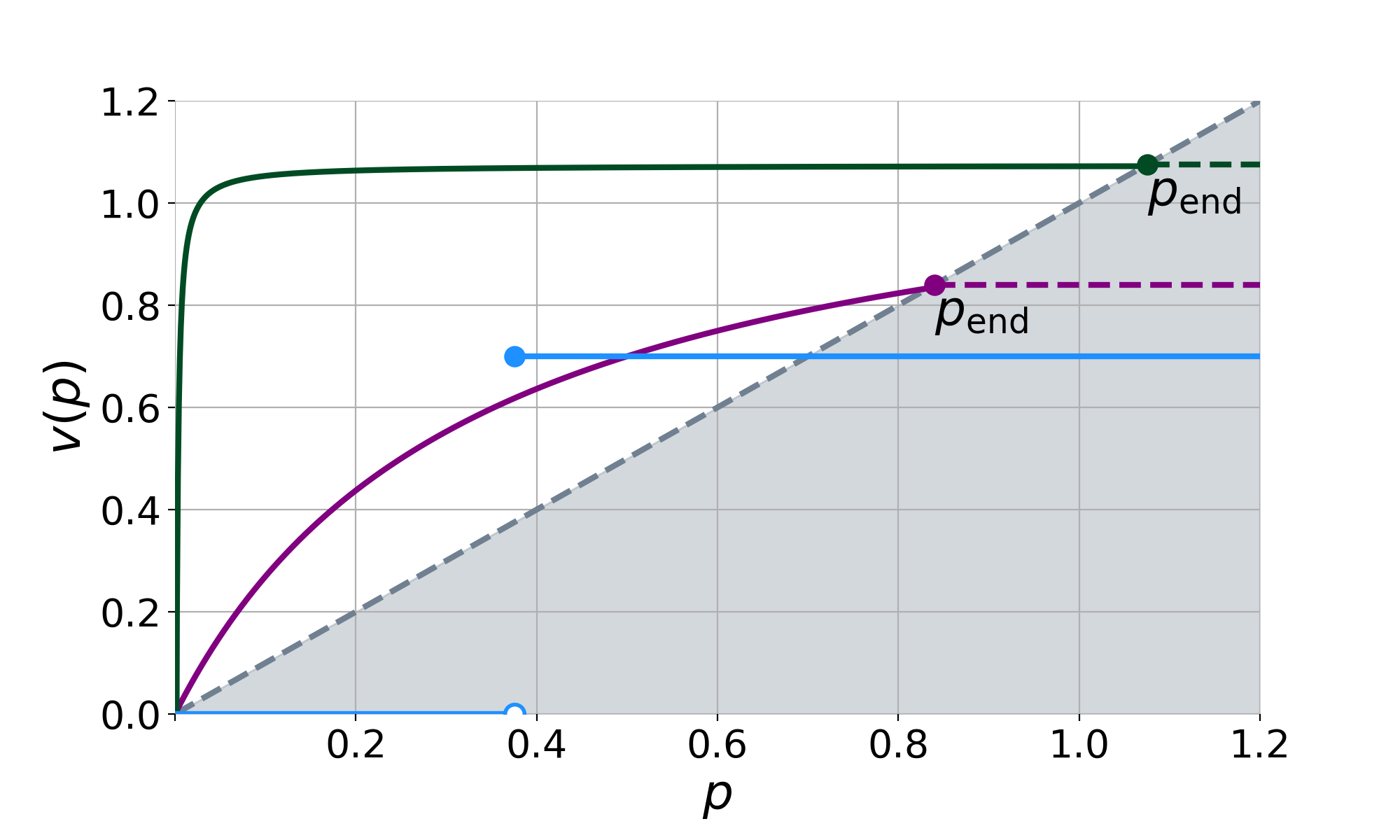}
  \caption{The purple, green, and blue lines are examples of value functions satisfying all of our assumptions for fully-sensitive, price-sensitive, and insensitive players respectively.} \label{fig:valuefunc}
\end{figure}

Casual games (e.g., AdVenture Capitalist), on the other hand, are generally simple, easy to learn, and easy to play. 
The play style of casual games reflects their simplicity: most play for a short time while in a line at a supermarket or waiting for a friend \cite{Turkowski_2023}. 
Designers of casual games generally do not incorporate in-game social communities and often obscure information about other players' purchasing habits such as whether they paid to accomplish tasks. 
Still, players can be very invested in casual games, spending large amounts time and money to progress \cite{gdcriseofidle2013, liftoff_2020}.
We assume these players do not know the type distribution induced by the given game but instead have a prior over what the type distribution might be. 
In this case, the true type distribution cannot affect player value, so a \emph{price-sensitive} player's value function depends only on price.

\begin{definition}\label{def:projval}
A price-sensitive player with type $\type$ at the beginning of a task has projected value $\type\valueFunc(\skippay)$ for finishing the task without skipping, when the skip price is $\skippay$.
\end{definition}

\noindent We assume the same properties about the relationship between value and price hold as just described.

Lastly, we look at players of hypercasual games (e.g., Cow Clicker, Cookie Clicker, Clicker Heroes).
Such players play many hypercasual games and are not invested in any particular game \cite{Turkowski_2023}.
We therefore consider it unlikely that players of hypercasual games are knowledgeable about the distribution of players for any given game.
Furthermore, hypercasual players predominantly play to relieve stress, pass time between daily
activities, or while multitasking, making it unlikely that they are reasoning about how price impacts their enjoyment \cite{facebook2019hyper}.
While skip prices may still need to be set high enough that the game does not feel pointless, beyond this threshold players value functions are \emph{insensitive} to skip prices.

\begin{definition}\label{def:projval_unsoph}
An insensitive player with type $\type$ at the beginning of a task has projected value $\valueFunc_c(\skippay) = \type\valueFunc$ when the skip price $\skippay > c$ for some constant $c$, and $0$ otherwise, for finishing the task without skipping. 
We assume $c$ is the same for all players. 
\end{definition}

For ease of exposition, the rest of this section will define player utilities with respect to price-sensitive values. 
The corresponding definitions for fully-sensitive and insensitive values are very similar and deferred to \Cref{app:util}.
\begin{definition}\label{def:margval}
At the beginning of a task, the marginal value a player with type $\type$ has for a skip at price $\skippay$ is
$$
    \valueFunc_m(\skippay) = (1-\type)\valueFunc(\skippay).
$$
The utility a player receives if they purchase a skip at price $\skippay$ is
$$
    \utilityFunc(\valueFunc(\skippay), \skippay) = \valueFunc(\skippay) - \skippay.
$$
We assume players are utility maximizers. 
Therefore, the utility a player receives in a task is
$$
    \utilityFunc_{\max}(\valueFunc(\skippay),\skippay,\type) = \max(\type\valueFunc(\skippay), \valueFunc(\skippay) - \skippay).
$$
\end{definition}

Our technical results in this paper leverage a common distributional condition from the economics literature called a \emph{monotone hazard rate} (MHR). 
\begin{definition}
A distribution $\dist(x)$ satisfies the MHR condition if the inverse hazard rate, $\frac{1-\dist(x)}{\dense(x)}$, is monotone non-increasing.
\end{definition}

\section{Optimal Pricing for Sensitive Players} \label{sec:util} 
We begin with studying price-sensitive players and show that they naturally give rise to the phenomenon of utility-optimal and revenue-optimal pricing coinciding at high prices. In particular, we show this when the game exhibits a \emph{whale distribution:} (1) a large majority of the population never buys anything; and (2) a small fraction of players generates a large majority of overall revenue.
Then we empirically demonstrate distributions where the results with price-sensitive players carry over to fully-sensitive players and others where their pricing behavior is different.
Note that while we will be characterizing the distributions that give rise to high prices, we only allow the game designer to control skip prices.
For the entirety of this section we focus on the setting with a single unrepeated task.

\subsection{Price-Sensitive Players}

The technical work of this section unfolds in three parts.
First, we show a condition on the type distribution under which it is utility optimal to set skip prices arbitrarily high and that it is satisfied by distributions with feature (1). 
Second, we give a characterization of how the value function can drive the revenue-optimal price higher and demonstrate the impact of distributions with feature (2). 
Finally, we leverage these two conditions to construct distributions where the revenue optimal payment still yields high utility.
In this section, for analytic convenience, we only consider value functions that are differentiable everywhere. 

We first need a description of the total utility generated by skips in a task.
\begin{proposition} \label{prop:util_0}
The expected utility (of all players) from a single task is
$$
    \expect[\type]{\utilityFunc_{\max}(\cdot)} = \typeDist\left(1-\frac{\skippay}{\valueFunc(\skippay)}\right)(\valueFunc(\skippay) - \skippay) + \left(1-\typeDist\left(1-\frac{\skippay}{\valueFunc(\skippay)}\right)\right)\expect[\type]{\type|\type >  1-\frac{\skippay}{\valueFunc(\skippay)}}\valueFunc(\skippay).
$$
We denote the utility-optimal price as
$$
    \skippay_{\mathrm{util}} = \argmax_{\skippay}\expect[\type]{\utilityFunc_{\max}(\cdot)},
$$
and the resulting utility it achieves as $\UTIL_{\max}.$
\end{proposition}

\begin{proof}
The first term accounts the utility contribution for players that purchase a skip.
The utility of a buyer is $\valueFunc(\skippay) - \skippay$ and the probability of sale is the probability their marginal value is greater than the price: $\prob[\type]{(1-\type)\valueFunc(\skippay) \geq p}$ and rearranging for $\type$ gives us $\prob[\type]{\type \leq 1-\frac{\skippay}{\valueFunc(\skippay)}} = \typeDist\left(1-\frac{\skippay}{\valueFunc(\skippay)}\right)$.
The second term is utility contribution from non-buyers.
That is, the expected discounted value that non-buyers receive, $\expect[\type]{\type|\type >  1-\nicefrac{\skippay}{\valueFunc(\skippay)}}\valueFunc(\skippay)$, multiplied by the probability they do not purchase a skip, $1-\typeDist\left(1-\frac{\skippay}{\valueFunc(\skippay)}\right)$.
\end{proof}
We begin by giving a condition implying that the utility-maximizing price is arbitrarily close to $\pnosale$. 
The following condition is of interest on its own, as it characterizes some cases when offering skip microtransactions becomes detrimental to expected player happiness.

\begin{theorem} \label{thm:util_nosale}
If
$$
    \typeDist\left(\frac{\valueFunc'(0) - 1}{\valueFunc'(0)}\right) \leq \frac{\valueFunc'(\pnosale)}{1-\valueFunc'(\pnosale)}\expect[\type]{\type},
$$
then $\forall \epsilon>0$, $\expect[\type]{\utilityFunc_{\max}(\valueFunc(\pnosale),\pnosale ,\type)} \geq \UTIL_{max}-\epsilon$. 
\end{theorem}

\begin{proof}
Define $\skippay^*$ to be $\argmax_{p_{0}\in [0,\pnosale-\epsilon)} \expect[\type]{\utilityFunc_{\max}(\valueFunc(\skippay),\skippay,\type)}$. 
Let $\startslope^*$ be the slope of the secant line between $\valueFunc(0)$ and $\valueFunc(\skippay^*)$ we can now write the value at $\skippay^*$ as $\valueFunc(\skippay^*) = \startslope^*\skippay^*$. 
From \Cref{prop:util_0}, the utility obtained at a point $\skippay^*<\pnosale$ is:
\begin{align*}
    \startslope^*\skippay^*\left(1-\typeDist\left(\frac{\startslope^*\skippay^* - \skippay^*}{\startslope^*\skippay^*}\right)\right)\expect[\type]{\type|\type \geq \frac{\startslope^*\skippay^*-\skippay^*}{\startslope^*\skippay^*}} + \skippay^*(\startslope^*-1)\typeDist\left(\frac{\startslope^*\skippay^*-\skippay^*}{\startslope^*\skippay^*}\right).
\end{align*}
Similarly, we define $\slopeX^*$ as the slope of secant line between the point $\valueFunc(\skippay^*)$ and $\valueFunc(\pnosale)$. We can write $\valueFunc(\pnosale) = \slopeX^*(\pnosale-\skippay^*)+\slope_1^*\skippay^*$.
Utilizing the fact that $\valueFunc(\skippay_x) = \skippay_x$ this simplifies to $\valueFunc(\pnosale)=\frac{\skippay^*(\startslope^*-\slopeX^*)}{1-\slopeX^*}$.
This gives utility of,
\begin{align*}
    \frac{\skippay^*(\startslope^*-\slopeX^*)}{1-\slopeX^*}\expect[\type]{\type}.
\end{align*}
Now, we can write the inequality for our condition:

\begin{align*}
    \skippay^*(\startslope^*-1)\typeDist\left(\frac{\startslope^*-1}{\startslope^*}\right)+\startslope^*\skippay^*\left(1-\typeDist\left(\frac{\startslope^*-1}{\startslope^*}\right) \right)\expect[\type]{\type|\type \geq \frac{\startslope^*-1}{\startslope^*}} &\leq \frac{\skippay^*(\startslope^*-\slopeX^*)}{1-\slopeX^*}\expect[\type]{\type}\\
    \typeDist\left(\frac{\startslope^*-1}{\startslope^*}\right)+ \frac{\startslope^*-1}{\startslope^*}\int_{\frac{\startslope^*-1}{\startslope^*}}^1 z\discountDense(z)\diff z  &\leq \frac{\startslope^*-\slopeX^*}{(1-\slopeX^*)(\startslope^*-1)}\expect[\type]{\type}.
\end{align*}
We now move the conditional expectation (i.e., the first term on the LHS) over to simplify the RHS.
\begin{align*}
    \typeDist\left(\frac{\startslope^*-1}{\startslope^*}\right) &\leq \frac{\startslope^*-\slopeX^*}{(1-\slopeX^*)(\startslope^* - 1)}\int_{0}^1z\discountDense(z)\diff z - \frac{\startslope^*}{(\startslope^* - 1)}\int_{\frac{\startslope^*-1}{\startslope^*}}^1z\discountDense(z)\diff z\\
    &= \frac{\startslope^*-\slopeX^*}{(1-\slopeX^*)(\startslope^*-1)}\int_{0}^{\frac{\startslope^*-1}{\startslope^*}}z\discountDense(z)\diff z 
    - \frac{\startslope^*(1-\slopeX^*)-(\startslope^*-\slopeX^*)}{(1-\slopeX^*)(\startslope^*-1)}\int_{\frac{\startslope^*-1}{\startslope^*}}^1z\discountDense(z)\diff z\\\\
    &= \frac{\startslope^*-\slopeX^*}{(1-\slopeX^*)(\startslope^*-1)}\int_{0}^{\frac{\startslope^*-1}{\startslope^*}}z\discountDense(z)\diff z +\frac{\startslope^*\slopeX^*-\slopeX^*}{(1-\slopeX^*)(\startslope^*-1)}\int_{\frac{\startslope^*-1}{\startslope^*}}^1z\discountDense(z)\diff z\\
    &\geq \frac{\startslope^*\slopeX^*-\slopeX^*}{(1-\slopeX^*)(\startslope^*-1)}\int_{0}^{\frac{\startslope^*-1}{\startslope^*}}z\discountDense(z)\diff z+\frac{\startslope^*\slopeX^*-\slopeX^*}{(1-\slopeX^*)(\startslope^*-1)}\int_{\frac{\startslope^*-1}{\startslope^*}}^1z\discountDense(z)\diff z
\end{align*}
The last step follows from the fact that $\slopeX^* < 1$.
Thus, we have:
\begin{align*}
    \typeDist\left(\frac{\startslope^* - 1}{\startslope^*}\right) &\leq  \frac{\slopeX^*}{1-\slopeX^*}\int_{\frac{\startslope^*-1}{\startslope^*}}^1z\discountDense(z)\diff z + \frac{\slopeX^*}{1-\slopeX^*}\int_{0}^{\frac{\startslope^*-1}{\startslope^*}}z\discountDense(z)\diff z \\ 
    \typeDist\left(\frac{\startslope^* - 1}{\startslope^*}\right) &\leq \frac{\slopeX^*}{1-\slopeX^*}\expect[\type]{\type}
\end{align*}
Finally, by concavity we must have both $\valueFunc'(0) \geq \startslope^*$ and $\valueFunc'(\pnosale) \leq \slopeX^*$ therefore it is sufficient to have,
\begin{align*}
    \typeDist\left(\frac{\valueFunc'(0) - 1}{\valueFunc'(0)}\right) &\leq \frac{\valueFunc'(\pnosale)}{1-\valueFunc'(\pnosale)}\expect[\type]{\type}
\end{align*}
If this condition holds the utility optimal price can be at most $\epsilon$ away from $\pnosale$. 
Furthermore, the slope of the value function (and therefore the slope of the utility function) above the utility optimal point can be at most $1$. 
This is because utility is strictly increasing while $\valueFunc'(\skippay)>1$. 
Therefore, $\expect[\type]{\utilityFunc_{\max}(\valueFunc(\pnosale),\pnosale, \type)} \geq \UTIL_{max}-\epsilon$
\end{proof}

The proof for this theorem follows from rearranging the definition of optimal utility and can be found in the appendix.
It is useful to interpret this condition through the lens of feature (1) of whale distributions.
The left hand side of the condition represents the highest possible probability of sale across all non-zero prices. 
Meanwhile, the right hand side contains the expectation of $\type$ which grows as the population becomes more patient.
Whale distributions, given a fixed value function, make this condition easier to satisfy; as the number of non-paying players increases, the probability of sale on the left hand side decreases and the expected value on the right hand side increases. 
In Figure \ref{fig:util_opt_impatient}, we can see that moving weight towards the patient portion of the population can shift the utility-optimal price up towards $\pnosale$. 
The intuition behind this result is quite simple: when there are enough patient players in the population, more total utility is gained by raising the price to make them happy than by providing cheap skips as an easy out for impatient players.

We now have a condition on the type distribution that would imply that the utility-optimal game design sets a price with zero probability of sale. 
The next step is to characterize distributions for which the revenue-optimal price monetizes only a small fraction of players.
First, we describe the revenue generated by skips in a single task with respect to the distribution over impatience.
It is useful to keep in mind that a player with high $\type$ is a patient player, therefore the distribution over $1-\type$ is a distribution over ``impatience.'' 

\begin{figure}[h]
     \centering
     \begin{subfigure}[b]{0.45\textwidth}
         \centering
         \includegraphics[trim={0.45cm 0 19cm 0},clip,width=0.8\textwidth]{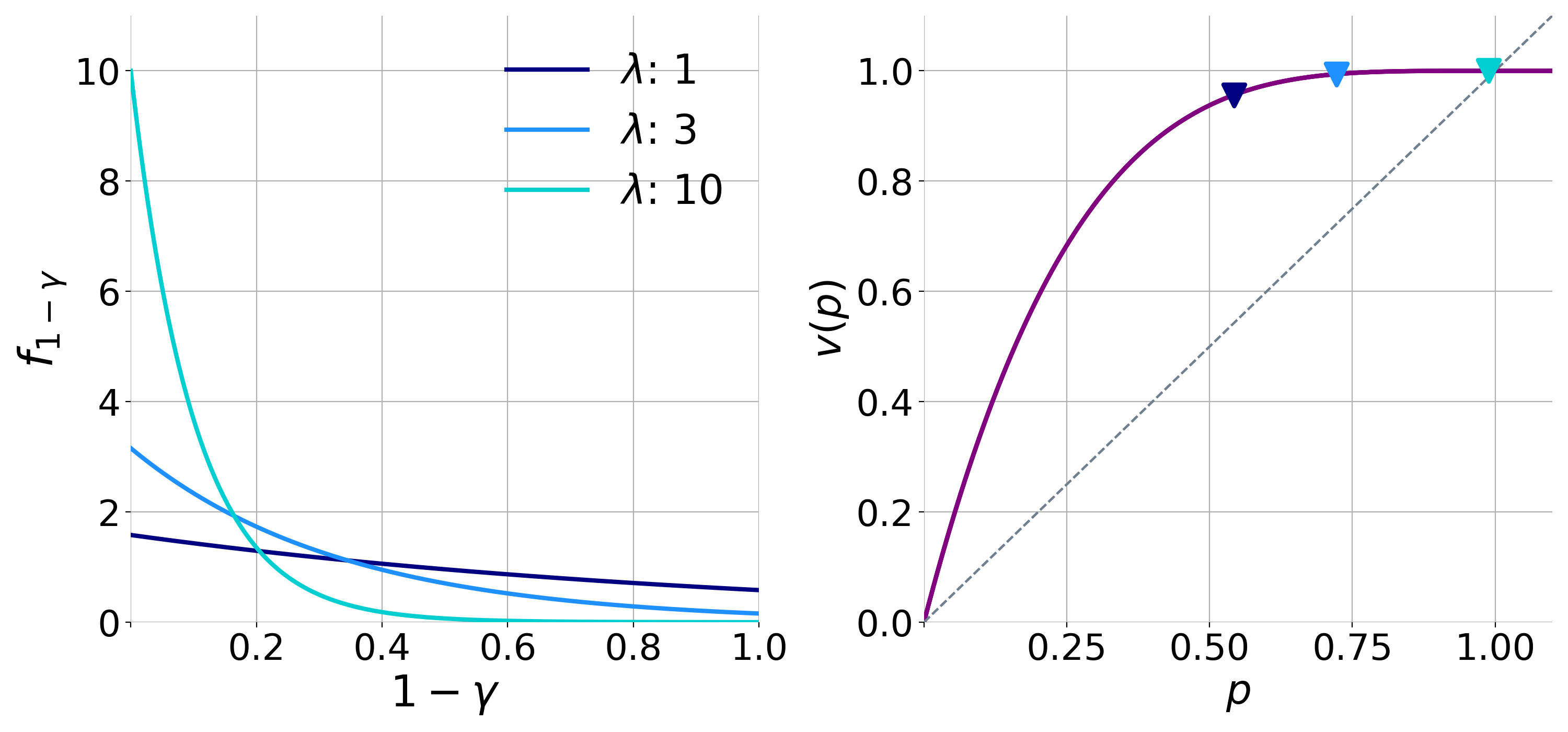}
         \caption{Probability density functions of an exponential distribution, with rate parameter $\lambda \in \{1, 3, 10\}$. As $\lambda$ increases so does the relative density of patient people.}
         \label{subfig:util_exp}
     \end{subfigure}\hfill%
     \begin{subfigure}[b]{0.45\textwidth}
         \centering
         \includegraphics[trim={19.5cm 0.5cm 0.55cm 0},clip,width=0.8\textwidth]{plots/patient_people_util_price.png}
         \caption{Utility optimal prices ({\color{navy} \isotri\blacktriangledown}, {\color{blue} \isotri\blacktriangledown}, {\color{lightblue} \isotri\blacktriangledown}) seen increasing with the relative density of patient players, plotted on $\valueFunc(\skippay) = 1-(\skippay-1)^4$.
         }
         \label{subfig:util_val}
     \end{subfigure}
     \caption{Utility Optimal Price as the Proportion of Patient People Increases}
     \label{fig:util_opt_impatient}
\end{figure}

\begin{proposition}\label{prop:rev_single}
    The revenue generated from offering a skip price $\skippay$ from a single task is
    $$
        \REV(\valueFunc(\skippay), \skippay, \type) = \skippay\left(1-\dist_{1-\type}\left(\frac{\skippay}{\valueFunc(\skippay)}\right)\right).
    $$
\end{proposition}

\begin{proof}
    For some price $\skippay$ the probability of sale is $\prob[\type]{(1-\type)\valueFunc(\skippay)\geq \skippay} = \prob[\type]{(1-\type)\geq \frac{\skippay}{\valueFunc(\skippay)}}$. 
    Thus, the expected revenue is $\skippay\left(1-\dist_{1-\type}\left(\frac{\skippay}{\valueFunc\left(\skippay\right)}\right)\right)$.
\end{proof}

\begin{lemma} \label{lma:con-rev}
The revenue-optimal price $\skippay_{\rev}\geq \skippay^*$ if for all $\skippay \leq \skippay^*$, 
$$
    \frac{1-\dist_{1-\type}\left(\frac{\skippay}{\valueFunc(\skippay)}\right)}{\dense_{1-\type}\left(\frac{\skippay}{\valueFunc(\skippay)}\right)} \geq \frac{\skippay}{\valueFunc(\skippay)}\left(1 - \frac{\skippay\valueFunc'(\skippay)}{\valueFunc(\skippay)}\right).
$$
Furthermore, if $\dist_{1-\type}$ satisfies MHR then the optimal price is the point where this is satisfied with equality.
\end{lemma}

\begin{proof}
We first take the derivative of revenue with respect to the price,
\begin{align*}
    \frac{\diff}{\diff\skippay} \REV(\valueFunc(\skippay), \skippay, \type)
    &= 1-\dist_{1-\type}\left(\frac{\skippay}{\valueFunc(\skippay)}\right) + \skippay\dense_{1-\type}\left(\frac{\skippay}{\valueFunc(\skippay)}\right)\frac{\skippay\valueFunc'(\skippay)-\valueFunc(\skippay)}{\valueFunc^2(\skippay)}
\end{align*}
We need this derivative to be positive along the entirety of the interval $[0, \skippay^*]$,
\begin{align*}
  \frac{1-\dist_{1-\type}\left(\frac{\skippay}{\valueFunc(\skippay)}\right)}{\dense_{1-\type}\left(\frac{\skippay}{\valueFunc(\skippay)}\right)} &\geq \skippay\frac{\valueFunc(\skippay)-\skippay\valueFunc'(\skippay)}{\valueFunc(\skippay)^2}\\
  & \geq \frac{\skippay}{\valueFunc(\skippay)}\left(1 - \frac{\skippay\valueFunc'(\skippay)}{\valueFunc(\skippay)}\right). \numberthis \label{align:HRutil}
\end{align*}
The second statement follows from the fact that 
$$
\frac{1-\dist_{1-\type}\left(\frac{\pnosale}{\valueFunc(\pnosale)}\right)}{\dense_{1-\type}\left(\frac{\pnosale}{\valueFunc(\pnosale)}\right)} = 0,
$$ 
and the term on the right hand side in (\ref{align:HRutil}) is both monotone increasing and approaches $0$ as $\skippay$ approaches $0$. 
Given that the term on the left hand side in (\ref{align:HRutil}) has MHR and therefore is also monotone, this maximization problem has a single unique solution. 
\end{proof}

Here, we interpret what this says for type distributions and their revenue-optimal payments. 
This lemma's first condition is a lower bound on the inverse hazard rate of the impatience distribution.
The larger the domain over which this holds, the larger lower bound we have for the revenue-optimal payment. 
As the density of the most impatient players increases, the domain where this condition is satisfied can only grow. 
One might worry that having the high density of patient non-buyers necessary for Theorem \ref{thm:util_nosale} could violate Lemma \ref{lma:con-rev}.
However, the density of players who would never buy are not considered directly by this condition; for large enough $\type$, $1-\type \neq \nicefrac{\skippay}{\valueFunc(\skippay)}$ for any price, so their density never appears in the inverse hazard rate.
Figure \ref{fig:rev_opt_impatient_flat} shows an example supporting this intuition: ignoring the empty circles for the moment, heavier tails drive up the revenue-optimal price while keeping the utility-optimal price high. 
While these distributions were hand-crafted for illustrative purposes, the same phenomena occur in more natural heavy-tailed distributions.

\begin{figure}[h]
     \centering
     \begin{subfigure}[b]{0.45\textwidth}
         \centering
         \includegraphics[trim={0.25cm 0 19cm 0},clip,width=0.8\textwidth]{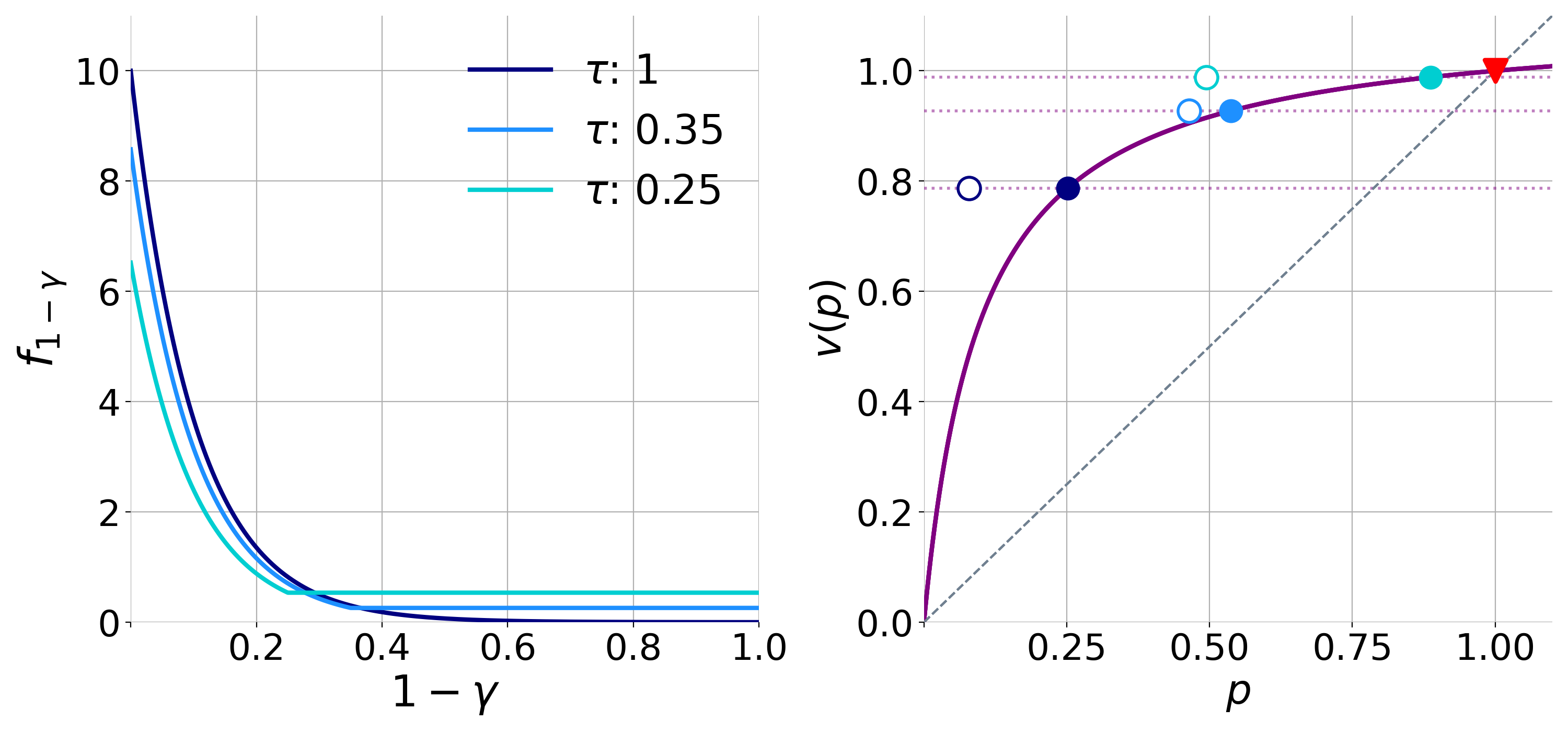}
         \caption{Probability density functions of an exponential distribution with rate parameter $\lambda = 10$ and density has been made a constant $\dense_{1-\type}(\tau)$ after $\tau$, with $\tau \in \{1, 0.35, 0.25\}$, and re-normalized.$\newline$$\newline$}
         \label{subfig:rev_exp_flat}
     \end{subfigure}\hfill%
     \begin{subfigure}[b]{0.45\textwidth}
         \centering
         \includegraphics[trim={19.5cm 0 0.5cm 0},clip,width=0.8\textwidth]{plots/impatient_people_rev_price_flat.png}
         \caption{Revenue optimal prices (\revcircle{0.5}{navy}{navy}, \revcircle{0.5}{blue}{blue}, \revcircle{0.5}{lightblue}{lightblue}) with respect to $\valueFunc(\skippay)$ seen increasing with the density of \emph{impatient} players while the utility optimal price ({\color{red} \isotri\blacktriangledown}) is held high. Revenue optimal prices (\revcircle{0.5}{white}{navy}, \revcircle{0.5}{white}{blue}, \revcircle{0.5}{white}{lightblue}) with respect to constant value functions $\in \{\valueFunc(\skippay_{\revcircle{0.2}{navy}{navy}}), \valueFunc(\skippay_{\revcircle{0.2}{blue}{blue}}),\valueFunc(\skippay_{\revcircle{0.2}{lightblue}{lightblue}})\}$ seen lower than non-constant counterparts.}
         \label{subfig:rev_val_flat}
     \end{subfigure}
     \caption{Revenue Optimal Price as the Proportion of Impatient People Increase}
     \label{fig:rev_opt_impatient_flat}
\end{figure}

The second part of the lemma gives us an analogue to virtual values.
Traditionally, virtual values are defined with respect to a constant value function. 
Suppose we define a constant value function as $\valueFunc = \valueFunc(\skippay^*)$ where $\skippay^*$ is some fixed point, the resulting virtual values will be greater for any point $\skippay^{**}>\skippay^*$ than defined with respect to $\valueFunc(\skippay^{**})$.
To see this, note that the virtual value for a constant value function would be identical except for the term $\left(1 - \nicefrac{\skippay\valueFunc'(\skippay)}{\valueFunc(\skippay)}\right)$ which is always less than one while the value function is increasing. 
This demonstrates the upward pressure an increasing value function has on pricing.  
In Figure \ref{fig:rev_opt_impatient_flat}, we see that in action: for a flat value function (the dotted purple lines) the revenue optimal strategy sets lower prices. 

Putting together the conditions from Theorem \ref{thm:util_nosale} and Lemma \ref{lma:con-rev}, we characterize a price that provides both high utility and high revenue.
\begin{corollary}\label{cor:util_rev_price}
Let $\skippay_{\rev}$ be the revenue optimal price and $\slope_{\rev}$ be the slope of the value function at $\valueFunc(\skippay_{\rev})$. 
If the condition in Theorem~\ref{thm:util_nosale} holds, then setting the revenue-optimal price $\skippay_{\rev}$ achieves utility $\expect[\type]{\utilityFunc_{\max}(\valueFunc(\skippay_{\rev}),\skippay_{\rev} ,\type)}\geq \UTIL_{\max}- \slope_{\rev}(\pnosale-\skippay_{\rev})$.
\end{corollary}

\begin{proof}
Given the condition in \Cref{thm:util_nosale}, the utility-optimal price is $\pnosale$. 
We upper bound the amount of utility loss by the distance between the two prices and the slope of the value function at the lower price.
\end{proof}

This simple fact demonstrates that we can ensure near-optimal utility by setting prices high enough to reach a near-flat region of the value function.
As mentioned in the beginning of the section, while high revenue-optimal and utility-optimal prices seem at odds with traditional mechanism design, many game designers are often encouraged to increase their prices for microtransactions and find increases in both overall retention and revenue \cite{levy_2016}. 

\begin{figure}[h]
     \centering
     \begin{subfigure}[b]{0.45\textwidth}
         \centering
         \includegraphics[trim={0.5cm 0 19cm 0},clip,width=0.8\textwidth]{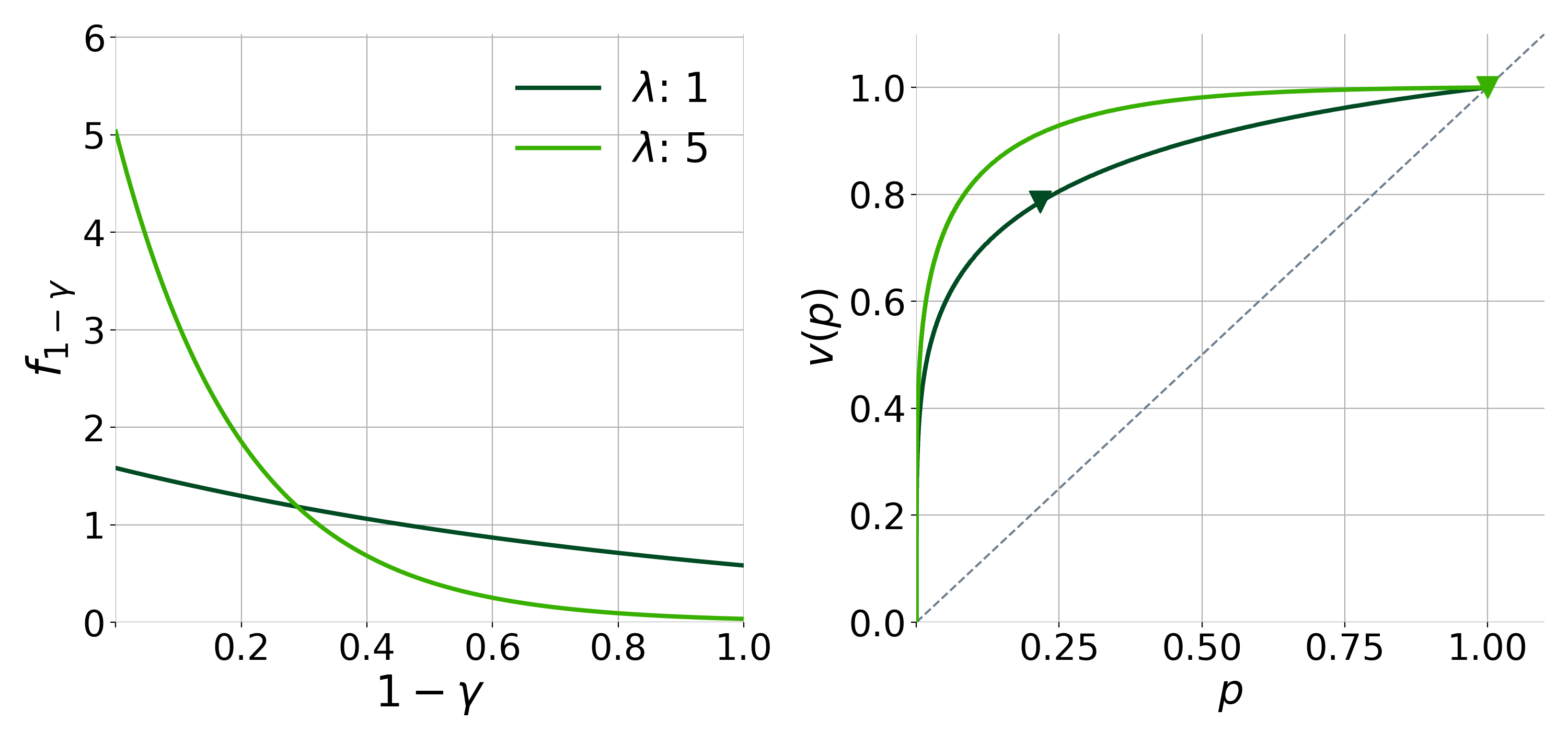}
         \caption{Probability density functions of an exponential distribution with rate parameter $\lambda \in \{1, 5\}$ as $\lambda$ increases so does the proportion of patient people.}
         \label{subfig:util_exp_inc_prsale}
     \end{subfigure}\hfill%
     \begin{subfigure}[b]{0.47\textwidth}
         \centering
         \includegraphics[trim={19cm 0 0.5cm 0},clip,width=0.8\textwidth]{plots/pr_sale_util_price_increase.png}
         \caption{Utility optimal prices ({\color{darkgreen} \isotri\blacktriangledown}, {\color{green} \isotri\blacktriangledown}) seen increasing with the density of patient players and plotted on $\valueFunc(\skippay)$'s generated from distributions in (a).}
         \label{subfig:util_val_inc_prsale}
     \end{subfigure}
     \newline
     \begin{subfigure}[b]{0.44\textwidth}
         \centering
         \includegraphics[trim={0.65cm 0 18.5cm 0},clip,width=0.8\textwidth]{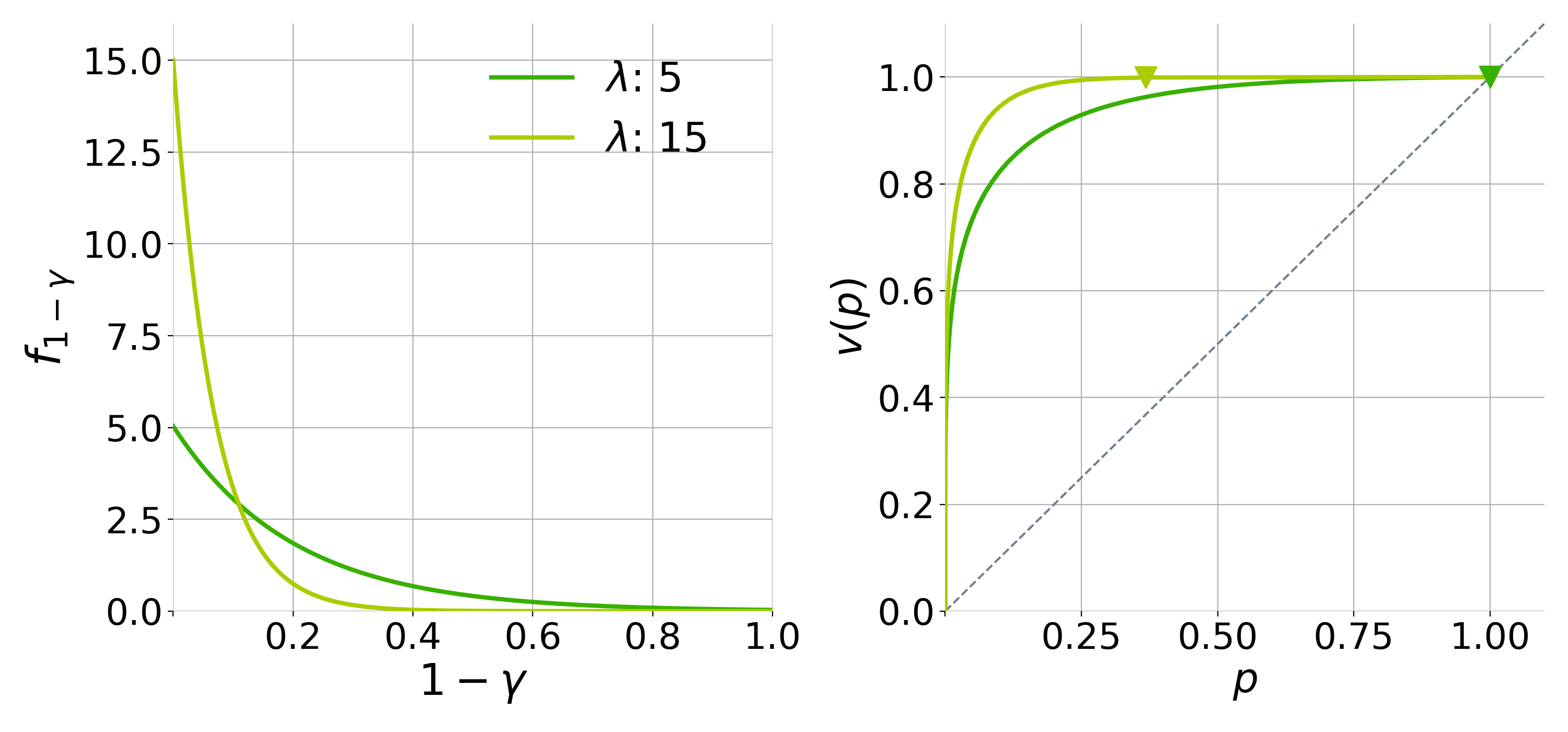}
         \caption{Probability density functions of an exponential distribution with $\lambda \in \{5, 15\}$ as $\lambda$ increases so does the proportion of patient people.}
         \label{subfig:util_exp_dec_prsale}
     \end{subfigure}\hfill%
     \begin{subfigure}[b]{0.45\textwidth}
         \centering
         \includegraphics[trim={20cm 0 0.5cm 0},clip,width=0.8\textwidth]{plots/pr_sale_util_price_decrease.png}
         \caption{Utility optimal prices ({\color{green} \isotri\blacktriangledown}, {\color{lightgreen} \isotri\blacktriangledown}) seen \emph{decreasing} with the density of patient players and plotted on $\valueFunc(\skippay)$'s generated from distributions in (a).}
         \label{subfig:util_val_dec_prsale}
     \end{subfigure}
     \caption{Utility Optimal Prices on Value Functions Dependent on Probability of Sale}
     \label{fig:vary_util_opt_price}
\end{figure}

\subsection{Fully-Sensitive Players}
The results of the previous section depend on both the value function and the type distribution. 
It is not obvious if these results hold when the value function itself depends on the type distribution. 
We now show empirically that for value functions that are linear in the probability of sale we can still get high utility and revenue optimal prices but sometimes optimal prices can be pushed down if the value function grows too slowly with price.

We first define a class of linear fully-sensitive value functions. 
For a fixed type distribution $\typeDist$ its corresponding $c$-linear value function is the function that is $0$ when $\skippay=0$ otherwise is the function that solves $c(1-\typeDist\left(1-\nicefrac{\skippay}{\valueFunc(\typeDist,\skippay)}\right)) = \valueFunc(\typeDist,\skippay)$. 
This ensures that values grow linearly as the probability of sale decreases but that the value's relationship with price is now dependent on the underlying type distribution. 
Given $\typeDist(x)$ is a fixed function with range $(0,1)$ we can simply solve for $\skippay$ for each value of $\valueFunc$ in range $(0,c)$. 
For all the distributions we consider, these value functions satisfy the desired properties listed in \Cref{sec:model}. 

\Cref{fig:vary_util_opt_price} shows how the utility-optimal price is affected by increasing the density of impatient players.
For a while, our intuition from the price-sensitive model holds (\Cref{subfig:util_val_inc_prsale}); create enough patient players and the utility-optimal price goes to the max. 
However, if there are too many patient players in the population, this trend flips and the utility-optimal price starts to drop (\Cref{subfig:util_val_dec_prsale}).
As distributions become more and more concentrated on patient players, the induced value functions plateau earlier.
In turn, the utility in this region for non-buyers remains relatively constant, while buyers' utility drops linearly with $\skippay$.
This makes sense since even a small price means few players will buy skips, and raising the price minimally affects the probability of sale. 
One way to view this result is that distributions can only become so patient before the induced value functions look insensitive to prices. 

In \Cref{fig:rev_val_inc_prsale}, we see a familiar trend: heavier tails usually mean higher revenue-optimal prices. 
However, it becomes more difficult to drive prices to the extremes in cases where the utility-optimal price is also high. 
This is because as the function gets flatter the upward pressure exerted on the price, as seen in \Cref{fig:rev_opt_impatient_flat}, is lessened.

\begin{figure}[h]
     \centering
     \begin{subfigure}[b]{0.45\textwidth}
         \centering
         \includegraphics[trim={0.5cm 0 19cm 0},clip,width=0.8\textwidth]{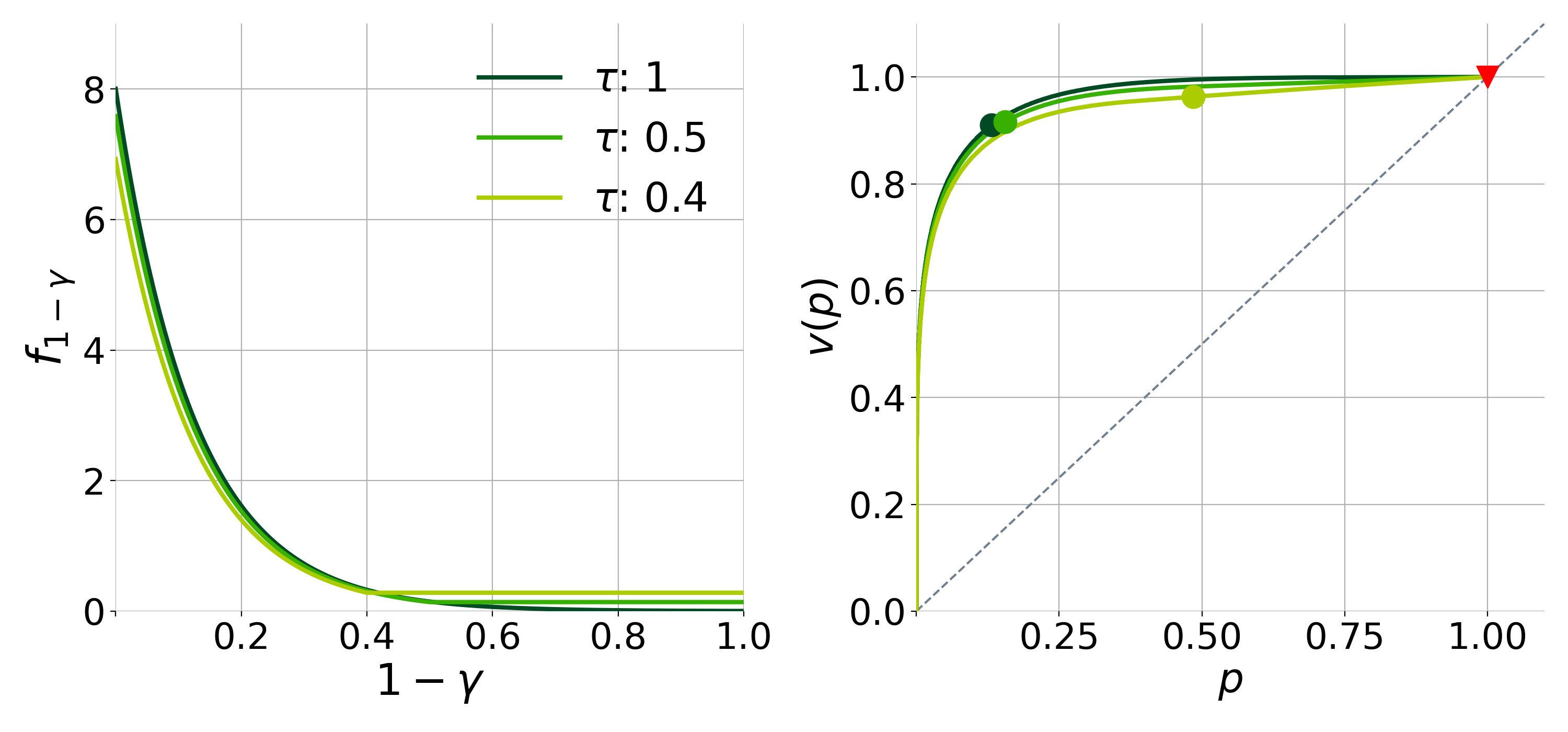}
         \caption{Probability density functions of an exponential distribution with rate parameter $\lambda = 8$ and density made a constant $\dense_{1-\type}(\tau)$ after $\tau$, with $\tau \in \{1, 0.5, 0.4\}$, and re-normalized.}
     \end{subfigure}\hspace{1cm}
     \begin{subfigure}[b]{0.45\textwidth}
         \centering
         \includegraphics[trim={19cm 0 0.5cm 0},clip,width=0.8\textwidth]{plots/pr_sale_rev_price_increase.png}
         \caption{Revenue optimal prices (\revcircle{0.5}{darkgreen}{darkgreen}, \revcircle{0.5}{green}{green}, \revcircle{0.5}{lightgreen}{lightgreen}) seen increasing with the density of impatient players plotted on $\valueFunc(\skippay)$'s generated from distributions in (a) while the utility optimal price ({\color{red} \isotri\blacktriangledown}) is held high.}
         \label{subfig:rev_val_inc_prsale}
     \end{subfigure}
     \caption{Revenue Optimal Prices when Values are Dependent on Probability of Sale.}
     \label{fig:rev_val_inc_prsale}
\end{figure}

This section shows that when value functions are sensitive to changes in price across the domain, the designer's job is relatively simple.
The utility-optimal price often dominates the revenue-optimal price, and so setting a high price and lowering it over time provides little risk in losing players. 
In fact, this is a strategy that game designers are sometimes advised to implement \cite{levy_2016}.
However, when value functions have large regions of insensitivity, as seen in \Cref{fig:pr_sale_flatregions}, it is not obvious how to trade off buyer utility and revenue beyond observing that one should not place the price in a high sensitivity region. 
This trade-off is the focus of the next section.

\begin{figure}[h]
     \centering
     \begin{subfigure}[b]{0.45\textwidth}
         \centering
         \includegraphics[trim={0.5cm 0 18.5cm 0},clip,width=0.8\textwidth]{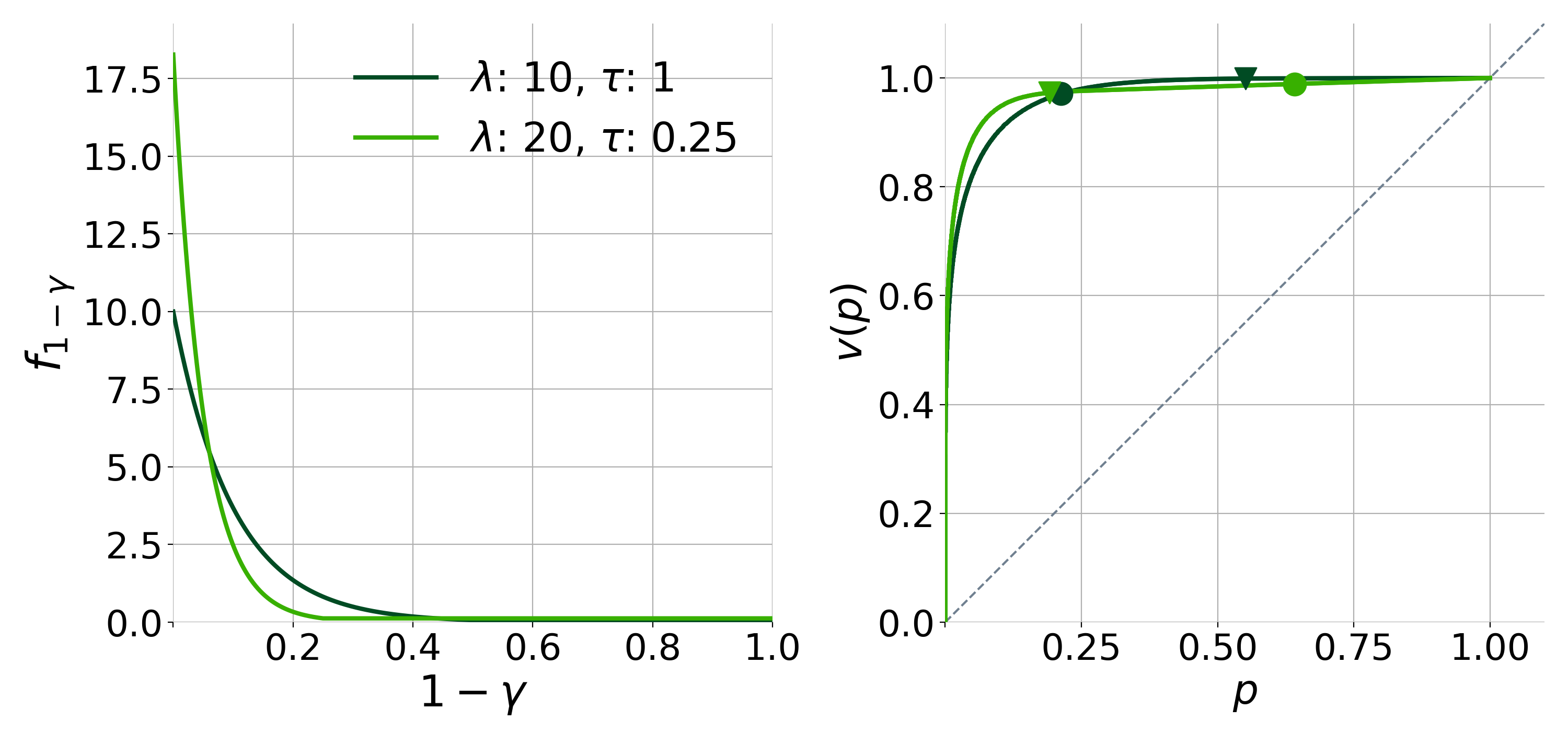}
         \caption{Probability density functions of exponential distributions with $(\lambda, \tau) \in \{(10, 1), (20, 0.25)\}$ with density $\dense_{1-\type}(\tau)$ after $\tau$ and re-normalized.}
     \label{subfig:pr_sale_flat_exp}
     \end{subfigure}\hfill%
     \begin{subfigure}[b]{0.46\textwidth}
         \centering
         \includegraphics[trim={20cm 0 0.5cm 0},clip,width=0.8\textwidth]{plots/pr_sale_flat_regions.png}
         \caption{Revenue optimal prices (\revcircle{0.5}{darkgreen}{darkgreen}, \revcircle{0.5}{green}{green}) seen increasing as the density of patient players concentrates and utility optimal prices ({\color{darkgreen} \isotri\blacktriangledown}, {\color{green} \isotri\blacktriangledown}) seen \emph{decreasing}.}
     \label{subfig:pr_sale_flat_util_low}
     \end{subfigure}
     \caption{Effect of Concentrating Patient Players on Revenue and Utility Optimal Prices.}
     \label{fig:pr_sale_flatregions}
\end{figure}

\section{Optimal Pricing for Insensitive Players}\label{sec:rev} 
We now tackle optimizing revenue for players insensitive to skip prices. 
We also note that we can transform fully- and price-sensitive value functions into an insensitive value function, allowing the results in this section to provide a good approximation when these sensitive value functions have large regions of low sensitivity. 
By definition, insensitive players have a minimum price below which their value is $0$. 
However, this point does not affect the analysis, so we assume it is at $0$ for the remainder of the section and let $\valueFunc \coloneqq \valueFunc_c(\skippay) $.
For more discussion see \Cref{appendix:cons_value}.

In this section, we leverage the analytic tractability afforded when players are insensitive to investigate the harder problem of repeated task pricing. 
Our main result is a simple pricing scheme that gives a constant approximation to the revenue of the optimal pricing scheme. 
This approximation ratio holds even if the designer is allowed to set multiple prices based on how much of the task players complete. 
Not only is our simple pricing scheme interesting theoretically, it is also practically relevant: we often see pricing schemes in which a single price is offered regardless of how much time has elapsed, even in mid-core (\emph{Clash of Clans}) and casual games (\emph{AdVenture Capitalist}). 
We end the section by testing our results in richer environments through simulations.

\subsection{Repeated Tasks}
We now introduce the further modeling assumptions and notation that we will need to study repeated tasks.
We make two main assumptions: player's types remain constant across tasks, and players are myopic to future value.
That is, they only consider the value of the task at hand rather than accounting for future value they might obtain by completing subsequent tasks. 
This matches experimental results on human behaviour showing that people often bracket their choices, optimizing piecemeal, rather than considering them jointly \cite{tversky1985framing}.
Mobile games often exacerbate this effect, obscuring the amount of time needed to complete subsequent tasks. 
For example, Clicker Heroes only shows the timers for current tasks, hiding all future wait timers. 

\begin{definition}
    At the beginning of each task $\task$, the designer sets a take-it-or-leave-it price $\skippay_{\task}$. This results in an infinite sequence of prices $\{\skippay_{\task}\}_{\task\in\mathbb{N}}$.
\end{definition}
\noindent One might imagine that game designers try to extract as much revenue from their players as possible. 
However, their task is not so simple. 
The mobile game marketplace is highly saturated with offerings in every conceivable permutation of aesthetic elements; e.g., \textit{Cow Clicker}, \textit{Cookie Clicker}, \textit{Planet Clicker}, \textit{Room Clicker}, \textit{Battery Clicker}, and \textit{Candy Clicker} are all real games!
In such a marketplace, players who are retained in one round may not be retained in the next; designers have to worry about losing players to these outside options at every round of the game. 
We characterize the appeal of these outside options as a retention threshold resampled from a static distribution at the beginning of every task.
This utility check means a game is always at a risk of losing its least satisfied players.
\begin{definition}
At the end of each task $\task \in [1, \infty)$, all players realize a retention threshold $\retention_{\task} \sim \retDist$, where $\supp(\retDist) = [0,\infty)$.
A player quits the game immediately if
$$
    \utilityFuncith{\task}_{\max}(\valueFunc(\skippay_{\task}),\skippay_{\task},\type)<\retention_{\task},
$$
where $\utilityFuncith{\task}_{\max}(\valueFunc(\skippay_{\task}),\skippay_{\task},\type) = \max(\type\valueFunc(\skippay_{\task}), \valueFunc(\skippay_{\task}) - \skippay_{\task})$.
\end{definition}

\subsection{An Upper Bound on Optimal Revenue}
Characterizing the optimal sequence of prices is difficult, especially when one must worry about the type distribution changing over time due to retention. 
However, the revenue generated when players' types are known to the designer is not only easy to characterize but upper bounds the revenue when types are unknown. 
Within this \emph{known types} setting, the optimal strategy sets $\skippay_i = (1-\type_{i})\valueFunc$ for each player $i$ which is guaranteed to sell and static across tasks.
We use the term \emph{unknown types} for the setting where the designer has only a prior over players' types and must price each round anonymously.
We assume designers exponentially discount future rewards characterized by $\gDiscount$.

\begin{proposition}\label{prop:rev_known}
The revenue generated from offering a static skip price $\skippay_i$ to a player $i$ with known type $\type_i$ and game designer discount $\gDiscount$ is
$
\REV(\valueFunc, \skippay_i, \type_i) = \frac{\skippay_i}{1-\gDiscount\retDist(\valueFunc-\skippay_i)}$, if $\skippay_i\leq(1-\type_i)\valueFunc$ and $0$ otherwise.
\end{proposition}

\begin{proof}
    When types are known, the only source of uncertainty to the designer is a player's retention threshold; the marginal value a player has for a skip, $\valueFunc_{m} = (1-\type)\valueFunc$, is fully known to the designer. 
    As a game designer can guarantee a sale, they will always set a price such that skipping is at least weakly preferred to waiting. 
    Furthermore, as $\type$'s do not change from task to task and $\retention$'s are redrawn i.i.d.\@ each task, task $2$ looks identical to task $1$ should the player remain in the game. 
    Therefore, the designer should choose a single static price to charge each player across tasks and we can write a player's projected utility for staying in a task as $\utilityFunc = \valueFunc - \skippay$. 
    
    Across multiple rounds, a player is in the game at task $\task$ if $\utilityFunc \geq \retention_{i}$ for all previous rounds $i\leq \task$.
    Summing over all rounds with designer discount $\gDiscount$ gives the expected revenue generated by that player:
    $$
    \REV(\valueFunc, \skippay, \type) = \begin{cases} 
          \skippay\sum_{i=0}^{\infty}\gDiscount^{i}\retDist(\valueFunc - \skippay)^i = \frac{\skippay}{1-\gDiscount\retDist(\valueFunc-\skippay)}, & \mathrm{if }\ \skippay\leq\valueFunc_{m}; \\
          0, & \mathrm{otherwise}.
       \end{cases}
    $$
\end{proof}

We now make an assumption on the retention distribution that is similar to regularity (many common distributions satisfy it e.g., uniform, exponential, Pareto).
\begin{definition}
    We say a retention distribution $\retDist(x)$ is \emph{retention regular} if $x-\frac{1-\gDiscount\retDist(\valueFunc-x)}{\gDiscount\retDens(\valueFunc-x)}$ is monotone non-decreasing.
\end{definition}\label{def:regularish}

\begin{proposition}\label{prop:pComplete}
The optimal price for a player $i$ given known types, under retention regularity of the retention distribution, is
$$
    \skippay_i=\min\left((1-\type_i)\valueFunc, 
    \frac{1-\gDiscount\retDist(\valueFunc-\skippay)}{\gDiscount\retDens(\valueFunc-\skippay)}
    \right),
$$
where $\skippay$ is the unique solution to the equation $q = \frac{1-\gDiscount\retDist(\valueFunc-q)}{\gDiscount\retDens(\valueFunc-q)}.$
\end{proposition}

\begin{proof}
    Retention regularity guarantees that there exists a unique maximizing solution.
    To find it, we take the derivative of the revenue expression above, 
    \begin{align*}
        \frac{\diff}{\diff\skippay}\REV(\valueFunc, \skippay, \type) &= \frac{1-\gDiscount\retDist(\valueFunc-\skippay) - \gDiscount\skippay\retDens(\valueFunc-\skippay)}{\left(1-\gDiscount\retDist(\valueFunc-\skippay)\right)^2} = 0 \\
        \skippay &= \frac{1-\gDiscount\retDist(\valueFunc-\skippay)}{\gDiscount\retDens(\valueFunc-\skippay)}.
    \end{align*}
    However, this calculation alone does not always find the optimal price: following such a prescription may set a price greater than a player's marginal value.
    As the designer has access to players' types, the seller can compare the price computed above with respect to the retention distribution and an agent's marginal value, setting the price to whichever is smaller.
\end{proof}

Notice that retention adds a constant upper bound on the price we should charge.
For players with high enough marginal value, charging a price that is too high may leave them with low utility, risking revenue from future tasks.
Conversely, if a player has low marginal value for a skip, full surplus extraction may still leave them with high enough utility to be retained.
We build on these observations for our simple pricing scheme later in this section.

Of course, designers are not restricted to offering only one skip per task, for example they could charge less if the player waits for half the task.
We show that setting multiple prices can generate more revenue in \Cref{appendix:complex_pricing}.
While even finding a closed form for such a pricing scheme is difficult, they are also upper bounded by the known types setting.

We note that without any distributional assumption on the marginal value distribution there is an unbounded revenue gap between known and unknown types, so we have no simple tool for bounding the performance of these complex pricing schemes.
\begin{theorem}\label{thm:unboundedgap}
For any constant $c$, there exists a retention distribution $\retDist$ and marginal value distribution $\margDist$ such that a game designer with known types achieves $c$ times more revenue than the revenue optimal price for unknown types.
\end{theorem}

\begin{proof}
    Let the value for the task be $e^c$. 
    We then define our retention distribution to take two possible values, $\retDens(0) = 0.5, \retDens(e^c+1)= 0.5$.
    Notice that such a characterization places retention completely outside of the power of the game designer. 
    Therefore, we have reduced our problem such that we need only consider the marginal value distribution to sell skips.
    We assume that the discount factors are distributed such that the marginal value distribution forms the \emph{equal-revenue distribution}, $\margDist(x) = 1- \nicefrac{1}{x}$.
    In the equal-revenue distribution setting any posted price generates revenue of $1$ but its expected value is only upper bounded by its support.
    With known types, the designer can set a price that obtains full surplus extraction equal to $\expect[x\sim \margDist]{x} = \ln{e^c} = c$. 
    However, without knowledge of a player's type the designer must set one fixed price ex-ante and for any fixed price the expected revenue of selling to an equal revenue player is $1$.
    Therefore, the best a game can achieve without knowledge of the discount factor is a $\nicefrac{1}{c}$ fraction of the optimal known types revenue.
\end{proof}

This theorem relies on reducing our setting to a well known setting in the literature known as the equal-revenue distribution \cite{hartline2009simple, hartline2013mechanism}. 
A natural question is whether there is a restricted class of distributions in which pricing schemes with unknown types can still perform competitively. 
Distributions with a monotone hazard rate (MHR) are one such example.

\subsection{Simple Skip Pricing}
The technical content unfolds in two steps: (1)~we bound the gap in revenue at a given round between the known types setting and a simple pricing scheme using a result from the literature on distributions with a monotone hazard rate (MHR); and (2)~we show our pricing scheme maintains at least as many players as the known types setting and always preserves MHR. 
Combining these two facts guarantees a simple pricing scheme can $\nicefrac{1}{e}$-approximate the optimal revenue.

We begin by defining revenue in the unknown types setting.
\begin{proposition}\label{prop:rev_unknown}
    The revenue generated with unknown types, from a vector of skip prices $\skippays$ and game designer discount $\gDiscount$ is
    $$
        \REV(\valueFunc, \skippays, \type) = \sum_{i=0}^{\infty} \gDiscount^i\skippays_i \margDistith{i}(\skippays_{i})
    $$
\end{proposition}
\begin{proof}
    With unknown types, the optimal pricing scheme does not set a constant price across rounds.
    However, should the player remain in the game task $\task$ looks identical to task $\task - 1$.
    From \Cref{def:cond_marg} we have that the distribution of players that are remaining in some task $\task$ is defined by $\margDistith{\task}(\skippays_{\task})$.
    Thus, summing over all rounds with designer discount $\gDiscount$ gives the expected revenue generated by all players: $ \REV(\valueFunc, \skippays, \type) = \sum_{i=0}^{\infty} \gDiscount^i\skippays_i \margDistith{i}(\skippays_{i})$.
\end{proof}

We now define how retention impacts the marginal value distribution across tasks in the unknown types setting.
Note that at the start of the game, the distribution over players' marginal values is $\margDist(\skippay) = \Pr_{\type}[(1-\type)\valueFunc \leq \skippay]$.
\begin{definition}\label{def:cond_marg}
    A marginal value distribution $\margDistith{\task}$ at task $\task > 0$ if a retention threshold of $\retention_{\task-1}$ was drawn at the end of task $\task-1$ is defined as:
    $ \margDistith{\task}(x) \coloneqq \margDistith{\task-1}(x \mid \utilityFunc_{\max}(\valueFunc(\skippay),\skippay_{\task-1},\type) \geq \retention_{\task-1})$, where $\margDistith{0}(x) \coloneqq \margDist(x).$
\end{definition}

We formalize the observation of ``threshold'' pricing from the known types setting for the unknown types setting.

\begin{definition}[Myerson Threshold (MT) Pricing]
First, the designer computes
$$
    \skipret = \frac{1-\gDiscount\retDist(\valueFunc-\skippay)}{\gDiscount\retDens(\valueFunc-\skippay)}.
$$
Then at each task $\task$, they compute the Myerson optimal posted price with respect to the marginal value distribution at that task:
$$
    \skippay_{\task}^{\myer} = \frac{1-\margDistith{\task}(x)}{\margDenseith{\task}(x)}.
$$
Finally, for task $\task$, the designer charges $\skippay_{\task} = \min\left(\skipret,\skippay_{\task}^{\myer}\right)$.
\end{definition}
\noindent Note that the $\skipret$ computation is the same as in Proposition \ref{prop:pComplete}.
We can think of this as the retention distribution setting a constant upper bound on the price across all tasks. 

We leverage the following result to lower-bound the revenue obtained by the Myerson optimal price.
\begin{lemma}\label{lma:MHR_lit}
    \cite[Theorem 4.37]{hartline2013mechanism} For any marginal value distribution $\margDistith{\task}$ satisfying the MHR property, the expected value is at most $e$ times the expected revenue generated in one round by setting the optimal posted price.
\end{lemma}
 
Finally, we show that as long as the marginal value distribution satisfies MHR at the first task, threshold pricing achieves a $\nicefrac{1}{e}$-approximation of the optimal revenue.

\begin{theorem}\label{thm:approx}
For every retention distribution $\retDist$ satisfying retention regularity and marginal value distribution $\margDist$ satisfying MHR, Myerson Threshold (MT) pricing generates a $\nicefrac{1}{e}$ fraction of the revenue in the known types setting.
\end{theorem}

\begin{proof}
    This proof has two main steps.
    We first show that for any marginal value distribution $\margDist$, MT pricing achieves a $\nicefrac{1}{e}$ approximation of the revenue of the known types setting in the first task and retains at least the same set of players. 
    We then show that MHR is preserved on the marginal value distributions at each task and therefore, inductively we achieve a $\frac{1}{e}$ fraction of the total revenue.
    Key to our proof strategy is that in the known types setting, the retention probability of any player that is charged their marginal value is the same had they not bought a skip: $\valueFunc - \margValue = \type\valueFunc$. 
    
    First, we look at the case where $\skippay^{\myer}_{1}\leq \skipret$ i.e., when the Myerson optimal price at task 1 is lower than the retention threshold price. 
    Following our prescription, MT pricing sets $\skippay=\skippay^{\myer}_{1}$.
    A trivial upper bound on the revenue of the known types setting at task $1$ is $\Ex[\type\sim \typeDist]{\margValue}$ and by \Cref{lma:MHR_lit} we know that the Myerson optimal price achieves a $\nicefrac{1}{e}$ portion of this revenue. 
    Furthermore, the utility obtained by each player with type $\type$ in the known type setting is $\max(\type\valueFunc, \valueFunc - \skipret)$ and under MT pricing they receive $\max(\type\valueFunc, \valueFunc - \skippay^{\myer})$.
    Given $\skippay^{\myer}\leq \skipret$ we have that $\forall \type$, $\max(\type\valueFunc, \valueFunc - \skippay^{\myer}) \geq \max(\type\valueFunc, \valueFunc - \skipret)$. 
    Therefore, every player who is retained in the known types setting is also retained by setting the Myerson optimal price. 
    
    The other case is when $\skipret\leq \skippay^{\myer}_{1}$. 
    In this case, MT pricing gives all players the same utility as they would receive in the known types case and so trivially their retention probabilities are the same. 
    What is left to show is that setting a price of $\skipret$ gets at least a $\nicefrac{1}{e}$ approximation of the revenue. 
    Observe that an upper bound on the revenue of the known types setting is $\skipret$ as that is the largest price it ever sets. 
    All we need is that the probability of sale of charging $\skipret$ is at least $\nicefrac{1}{e}$. 
    Here, we leverage another result from the literature that says the Myerson optimal price always achieves probability of sale at least $\nicefrac{1}{e}$ on MHR distributions (Lemma 4.1 in \cite{hartline2008optimal}).
    Therefore, since $\skipret\leq \skippay^{\myer}_{1}$ it must have probability of sale at least $\nicefrac{1}{e}$ and revenue at at least $\nicefrac{\skipret}{e}$.
    
    Our inductive step simply needs to show that MHR is maintained for all subsequent marginal value distributions $\margDistith{\task>1}$.
    Observe that for any fixed price $\skippay_{\task}$ at task $\task$, $\utilityFuncith{\task}_{\max}(\valueFunc(\skippay_{\task}),\skippay_{\task},\type) = \max(\type\valueFunc(\skippay_{\task}), \valueFunc(\skippay_{\task})-\skippay_{\task})$ is weakly increasing in $\type$. 
    Therefore, for any realization of $\retention_{\task}$ there is some $\type^*$ such that all players with $\type< \type^*$ leave the game and all players with types $\type\geq\type^*$ are retained. 
    This means that $\margDistith{\task+1}$ is simply $\margDistith{\task}$ right truncated at the point $(1-\type^*)\valueFunc$ and by \Cref{lma:truncMHR} MHR will be preserved.
\end{proof}

Note that for any discount factor, marginal values decrease as the task goes on. 
Since in the known-type case the designer can guarantee a sale, they will always sell in the first round. 
This means the upper bound used in the proof of \Cref{thm:approx} applies equally to the multiple prices setting.
\begin{corollary}\label{cor:single_multi_bound}
For every retention distribution $\retDist$ satisfying retention regularity and marginal value distribution $\margDist$ satisfying MHR, Myerson Threshold (MT) pricing generates a $\nicefrac{1}{e}$ fraction of the revenue of any complex pricing scheme.
\end{corollary}
\begin{proof}
This corollary follows from the fact that the optimal known type mechanism is the same in the problem with a single skip price and where multiple prices can be set. 
To see this note that for any discount factor, marginal values decrease as the task goes on. 
Since in the known type case the designer can guarantee a sale they will always sell in the first round. 
This means the upper bound used in the proof of \Cref{thm:approx} applies equally to the multiple prices setting.
\end{proof}

\subsection{Simulations}
In this section, we simulate populations of players to study what features of the problem make MT pricing perform well. 
We also test the robustness of MT pricing by adding new ingredients that break away from some of our previous assumptions.
Namely, we allow the population to grow from newly acquired players and players to be heterogeneous with respect to retention draws. 
Full details of the parameters used in the simulations and the formal descriptions of these new ingredients can be found in \Cref{appendix:simparams}.
We note that even when adding these new ingredients, MT pricing gets at least a $\nicefrac{1}{e}$ fraction of the optimal known types' revenue.
Therefore, we compare MT pricing's performance to other simple pricing schemes to learn what features of the problem makes each option perform well. 
Histograms of the relative performance of MT pricing and the next best alternative for each of the following sections can be found in \Cref{appendix:images}. 

\subsubsection{Comparison to Myerson and Retention Threshold Pricing}
We begin by comparing two alternative pricing strategies: at each round, charge only the Myerson optimal price or the retention threshold price. 
The results (see \Cref{fig:alt_pricing_revs}) suggest that the lower price often achieves the higher revenue. 
Setting a retention threshold price achieves higher revenue in over 20\% of the instances and most of these instances are when Myerson takes a risk in setting a high price and loses too many players.

This matches the advice we see from the gaming community.
Rather than losing too many players trying to make a quick buck, MT pricing keeps enough players around to monetize in the next task. 
In fact, the revenue generated from MT pricing is within $1$\% of the revenue from the max of the two other pricing schemes in $82$\% of cases and strictly outperforms the next best option in $12$\% of cases. 
 
\subsubsection{Viral Population Growth}
The gap in performance between MT pricing and the alternatives is small in most settings.
This should not be surprising; MT charges the minimum of Myerson and the retention threshold and in many cases this results in exactly the same strategy as one of them, including retaining the same individuals.

\subsubsection{Independently Drawn Retention}
Once again, the strategy that sets the lower price often generates higher average revenue. 
However, Myerson pricing now consistently loses players when it gambles on a high price, eliminating any extra gains it could have achieved. 

\subsubsection{Lowering Payments}
The main finding of our simulations thus far is that lower prices are usually better.
One might wonder if the MT pricing algorithm could be improved by charging $\skippay_{\task} = \min\left(\skipret,c\skippay_{\task}^{\myer}\right)$ where $c$ is a constant less than one.
We re-examine a subset of previous settings that showed evidence of lower prices being performant and find that decreasing the scaling factor can be relatively safe for large enough values of $c$.
However, for values lower than $\nicefrac{1}{2}$, the variance in the performance not only increases but average performance decreases. 

\section{Conclusion} \label{sec:future}
Mobile games give rise to various phenomena that are difficult or impossible to explain with standard economic models. 
This paper aims to explain one such phenomenon: the paradoxical relationship between utility and revenue when selling skips.
We find that revenue-optimal pricing requires carefully balancing player utility within a task and across tasks; pricing not only needs to be high for players to have any value in game-playing but low enough to monetize them long-term.
We believe that we have only scratched the surface of a rich economic domain; there are many other phenomena arising in mobile games (and beyond) upon which our behavioural model could shed light. 
Here are a few examples.

\subsection{Battle passes}
With the success of skips, grind-based monetization has bled into all categories of gaming.
Battle passes are a noteworthy example; most of the top grossing console games have battle passes.
They are essentially play-time timers that reward the players at regular intervals with cosmetic items that have no affect on gameplay. 
Similarly to mobile games, players can pay money to skip past these timers. 

While the inner workings of a battle pass fit squarely within our model, interestingly, most battle passes require players to pay an entrance fee to participate in the grind. 
This cost is usually much lower then the cost to purchase cosmetics outright.
For example, in 2020 the price of a battle pass in Clash of Clans was $\$4.99$ USD, around $25$ times cheaper than other progression content \cite{joas_2020}. 
Cheap entrance fees are low risk ways of monetizing low value players as they still need to grind to complete tasks and earn rewards; an item's value is still anchored by the skip price for that task.
This adds another dimension to the problem where the game designer must determine how much to charge to enter knowing that if a player buys they will potentially provide more revenue and are retained longer by purchasing skips \cite{joas_2020}.

Furthermore, players who purchase often have higher retention from the value the grind provides \cite{joas_2020}.

\subsection{Changing Values}
While we allow players to differ in impatience, one simplification we made is that they all have the same underlying value for tasks. 
Such an assumption focuses on values derived primarily through effort.
When completion also rewards players with cosmetic items, players' preferences need not be aligned. 
Characterizing a player's aesthetic preferences by a second type would enable our model to explain players having changing preferences over skipping and waiting across tasks. 

In addition to each player having different values for items there is evidence to suggest these values may evolve overtime.
We see in empirical data that purchase intention can depend on the total time a player has spent in the game \cite{gdclongterm2013}.
This might motivate a player value function that grows over the course of the game.
Evolving value would further justify the safe approach to revenue maximization in \Cref{sec:rev}; players are allowed to play without spending money until their value grows to the point at which they are motivated to pay. 

\section*{Acknowledgements}
This work was funded by an NSERC Discovery Grant, a DND/NSERC Discovery Grant Supplement, a CIFAR Canada AI Research Chair (Alberta Machine Intelligence Institute), awards from Facebook Research and Amazon Research, and DARPA award FA8750-19-2-0222, CFDA \#12.910 (Air Force Research Laboratory).

\bibliographystyle{ACM-Reference-Format}
\bibliography{ref}


\begin{thebibliography}{54}


\ifx \showCODEN    \undefined \def \showCODEN     #1{\unskip}     \fi
\ifx \showDOI      \undefined \def \showDOI       #1{#1}\fi
\ifx \showISBNx    \undefined \def \showISBNx     #1{\unskip}     \fi
\ifx \showISBNxiii \undefined \def \showISBNxiii  #1{\unskip}     \fi
\ifx \showISSN     \undefined \def \showISSN      #1{\unskip}     \fi
\ifx \showLCCN     \undefined \def \showLCCN      #1{\unskip}     \fi
\ifx \shownote     \undefined \def \shownote      #1{#1}          \fi
\ifx \showarticletitle \undefined \def \showarticletitle #1{#1}   \fi
\ifx \showURL      \undefined \def \showURL       {\relax}        \fi
\providecommand\bibfield[2]{#2}
\providecommand\bibinfo[2]{#2}
\providecommand\natexlab[1]{#1}
\providecommand\showeprint[2][]{arXiv:#2}

\bibitem[Activision(2022)]%
        {activision_2021}
\bibfield{author}{\bibinfo{person}{Activision}.}
  \bibinfo{year}{2022}\natexlab{}.
\newblock \bibinfo{title}{2021 Annual Report - Activision Blizzard | Investor
  relations}.
\newblock
\newblock
\urldef\tempurl%
\url{https://investor.activision.com/static-files/d7b4f08d-213b-4bd5-a41b-7497baa9c106}
\showURL{%
\tempurl}


\bibitem[AppsFlyer(2020)]%
        {liftoff_2020}
\bibfield{author}{\bibinfo{person}{AppsFlyer}.}
  \bibinfo{year}{2020}\natexlab{}.
\newblock \bibinfo{title}{2020 Mobile Gaming Apps Report}.
\newblock
\newblock
\urldef\tempurl%
\url{https://liftoff.io/resources/report/2020-mobile-gaming-apps/}
\showURL{%
\tempurl}


\bibitem[Battalio et~al\mbox{.}(1991)]%
        {giffen1991}
\bibfield{author}{\bibinfo{person}{Raymond Battalio}, \bibinfo{person}{John
  Kagel}, {and} \bibinfo{person}{Kogut Carl}.} \bibinfo{year}{1991}\natexlab{}.
\newblock \showarticletitle{Experimental Confirmation of the Existence of a
  Giffen Good}.
\newblock \bibinfo{journal}{\emph{The American Economic Review}}
  \bibinfo{volume}{81}, \bibinfo{number}{4} (\bibinfo{year}{1991}),
  \bibinfo{pages}{961--970}.
\newblock


\bibitem[Brightman(2016)]%
        {brightman_2016}
\bibfield{author}{\bibinfo{person}{James Brightman}.}
  \bibinfo{year}{2016}\natexlab{}.
\newblock \bibinfo{title}{Only 3.5\% of gamers make in-app purchases -
  appsflyer}.
\newblock
\newblock
\urldef\tempurl%
\url{https://www.gamesindustry.biz/only-3-5-percent-of-gamers-make-in-app-purchases-appsflyer}
\showURL{%
\tempurl}


\bibitem[Cai et~al\mbox{.}(2019)]%
        {cai2019purchases}
\bibfield{author}{\bibinfo{person}{Jie Cai}, \bibinfo{person}{Donghee~Yvette
  Wohn}, {and} \bibinfo{person}{Guo Freeman}.} \bibinfo{year}{2019}\natexlab{}.
\newblock \showarticletitle{Who Purchases and Why? Explaining Motivations for
  In-Game Purchasing in the Online Survival Game Fortnite}.
\newblock \bibinfo{journal}{\emph{Proceedings of the Annual Symposium on
  Computer-Human Interaction in Play}} (\bibinfo{year}{2019}),
  \bibinfo{pages}{391–396}.
\newblock


\bibitem[Carter and Bj{\"o}rk(2016)]%
        {carter2016cheating}
\bibfield{author}{\bibinfo{person}{Marcus Carter} {and}
  \bibinfo{person}{Staffan Bj{\"o}rk}.} \bibinfo{year}{2016}\natexlab{}.
\newblock \showarticletitle{Cheating in candy crush saga}.
\newblock \bibinfo{journal}{\emph{Social, casual and mobile games: the changing
  gaming landscape}} (\bibinfo{year}{2016}), \bibinfo{pages}{261--74}.
\newblock


\bibitem[Chen et~al\mbox{.}(2021)]%
        {chen2019loot}
\bibfield{author}{\bibinfo{person}{Ningyuan Chen}, \bibinfo{person}{Adam~N
  Elmachtoub}, \bibinfo{person}{Michael~L Hamilton}, {and}
  \bibinfo{person}{Xiao Lei}.} \bibinfo{year}{2021}\natexlab{}.
\newblock \showarticletitle{Loot box pricing and design}.
\newblock \bibinfo{journal}{\emph{Management Science}} \bibinfo{volume}{67},
  \bibinfo{number}{8} (\bibinfo{year}{2021}), \bibinfo{pages}{4809--4825}.
\newblock


\bibitem[Choi and Kim(2004)]%
        {choi2004people}
\bibfield{author}{\bibinfo{person}{Dongseong Choi} {and}
  \bibinfo{person}{Jinwoo Kim}.} \bibinfo{year}{2004}\natexlab{}.
\newblock \showarticletitle{Why people continue to play online games: In search
  of critical design factors to increase customer loyalty to online contents}.
\newblock \bibinfo{journal}{\emph{CyberPsychology \& behavior}}
  \bibinfo{volume}{7}, \bibinfo{number}{1} (\bibinfo{year}{2004}),
  \bibinfo{pages}{11--24}.
\newblock


\bibitem[Crux(2023)]%
        {pixel_crux}
\bibfield{author}{\bibinfo{person}{Pixel Crux}.}
  \bibinfo{year}{2023}\natexlab{}.
\newblock \bibinfo{title}{Clash of Clans Village Progress Tracker}.
\newblock
\newblock
\urldef\tempurl%
\url{https://pixelcrux.com/Clash_of_Clans/Village_Progress_Tracker/App}
\showURL{%
\tempurl}


\bibitem[Dierks and Seuken(2020)]%
        {dierks2020revenue}
\bibfield{author}{\bibinfo{person}{Ludwig Dierks} {and} \bibinfo{person}{Sven
  Seuken}.} \bibinfo{year}{2020}\natexlab{}.
\newblock \showarticletitle{Revenue Maximization for Consumer Software:
  Subscription or Perpetual License?}
\newblock \bibinfo{journal}{\emph{arXiv preprint arXiv:2007.11331}}
  (\bibinfo{year}{2020}).
\newblock


\bibitem[Engblom(2017)]%
        {engblom_2017}
\bibfield{author}{\bibinfo{person}{Stefan Engblom}.}
  \bibinfo{year}{2017}\natexlab{}.
\newblock \bibinfo{title}{Balancing Cards in Clash Royale}.
\newblock
\newblock
\urldef\tempurl%
\url{https://www.youtube.com/watch?v=bHLQQh8Ctu4&amp;t=969s}
\showURL{%
\tempurl}


\bibitem[Evans(2016)]%
        {evans2016economics}
\bibfield{author}{\bibinfo{person}{Elizabeth Evans}.}
  \bibinfo{year}{2016}\natexlab{}.
\newblock \showarticletitle{The economics of free: Freemium games, branding and
  the impatience economy}.
\newblock \bibinfo{journal}{\emph{Convergence}} \bibinfo{volume}{22},
  \bibinfo{number}{6} (\bibinfo{year}{2016}), \bibinfo{pages}{563--580}.
\newblock


\bibitem[Facebook(2019)]%
        {facebook2019hyper}
\bibfield{author}{\bibinfo{person}{Facebook}.} \bibinfo{year}{2019}\natexlab{}.
\newblock \bibinfo{title}{Genre and great games report}.
\newblock
\newblock
\urldef\tempurl%
\url{https://www.facebook.com/fbgaminghome/marketers/build-great-games/game-design-hub/genre-great-games-report}
\showURL{%
\tempurl}
\newblock
\shownote{Accessed:2023-07-14}.


\bibitem[Fields and Cotton(2011)]%
        {fields2011social}
\bibfield{author}{\bibinfo{person}{Tim Fields} {and} \bibinfo{person}{Brandon
  Cotton}.} \bibinfo{year}{2011}\natexlab{}.
\newblock \bibinfo{booktitle}{\emph{Social game design: Monetization methods
  and mechanics}}.
\newblock \bibinfo{publisher}{CRC Press}.
\newblock


\bibitem[Geng and Chen(2019)]%
        {geng2019conspicuous}
\bibfield{author}{\bibinfo{person}{Wei Geng} {and} \bibinfo{person}{Zuguang
  Chen}.} \bibinfo{year}{2019}\natexlab{}.
\newblock \showarticletitle{Optimal Pricing of Virtual Goods with Conspicuous
  Features in a Freemium Model}.
\newblock \bibinfo{journal}{\emph{International Journal of Electronic
  Commerce}} \bibinfo{volume}{23}, \bibinfo{number}{3} (\bibinfo{year}{2019}),
  \bibinfo{pages}{427--449}.
\newblock


\bibitem[Goetz and Lu(2022)]%
        {goetz2022peer}
\bibfield{author}{\bibinfo{person}{Daniel Goetz} {and} \bibinfo{person}{Wei
  Lu}.} \bibinfo{year}{2022}\natexlab{}.
\newblock \showarticletitle{Peer Effects from Friends and Strangers: Evidence
  from Random Matchmaking in an Online Game}. In
  \bibinfo{booktitle}{\emph{Proceedings of the 23rd ACM Conference on Economics
  and Computation}}. \bibinfo{pages}{285--286}.
\newblock


\bibitem[Greer(2013)]%
        {Greer_2013}
\bibfield{author}{\bibinfo{person}{Emily Greer}.}
  \bibinfo{year}{2013}\natexlab{}.
\newblock \bibinfo{title}{Building games for the long-term: Maximizing
  monetization and player satisfaction}.
\newblock
\newblock
\urldef\tempurl%
\url{https://www.gdcvault.com/play/1018265/Building-Games-for-the-Long}
\showURL{%
\tempurl}


\bibitem[Grieve et~al\mbox{.}(2013)]%
        {grieve2013face}
\bibfield{author}{\bibinfo{person}{Rachel Grieve}, \bibinfo{person}{Michaelle
  Indian}, \bibinfo{person}{Kate Witteveen}, \bibinfo{person}{G~Anne Tolan},
  {and} \bibinfo{person}{Jessica Marrington}.} \bibinfo{year}{2013}\natexlab{}.
\newblock \showarticletitle{Face-to-face or Facebook: Can social connectedness
  be derived online?}
\newblock \bibinfo{journal}{\emph{Computers in human behavior}}
  \bibinfo{volume}{29}, \bibinfo{number}{3} (\bibinfo{year}{2013}),
  \bibinfo{pages}{604--609}.
\newblock


\bibitem[Hanner and Zarnekow(2015)]%
        {hanner2015purchasing}
\bibfield{author}{\bibinfo{person}{Nicolai Hanner} {and}
  \bibinfo{person}{Ruediger Zarnekow}.} \bibinfo{year}{2015}\natexlab{}.
\newblock \showarticletitle{Purchasing behavior in free to play games: Concepts
  and empirical validation}. In \bibinfo{booktitle}{\emph{2015 48th Hawaii
  International Conference on System Sciences}}. IEEE,
  \bibinfo{pages}{3326--3335}.
\newblock


\bibitem[Hartline et~al\mbox{.}(2008)]%
        {hartline2008optimal}
\bibfield{author}{\bibinfo{person}{Jason Hartline}, \bibinfo{person}{Vahab
  Mirrokni}, {and} \bibinfo{person}{Mukund Sundararajan}.}
  \bibinfo{year}{2008}\natexlab{}.
\newblock \showarticletitle{Optimal marketing strategies over social networks}.
  In \bibinfo{booktitle}{\emph{Proceedings of the 17th international conference
  on World Wide Web}}. \bibinfo{pages}{189--198}.
\newblock


\bibitem[Hartline(2013)]%
        {hartline2013mechanism}
\bibfield{author}{\bibinfo{person}{Jason~D Hartline}.}
  \bibinfo{year}{2013}\natexlab{}.
\newblock \showarticletitle{Mechanism design and approximation}.
\newblock \bibinfo{journal}{\emph{Book draft. October}} (\bibinfo{year}{2013}),
  \bibinfo{pages}{145}.
\newblock


\bibitem[Hartline and Roughgarden(2009)]%
        {hartline2009simple}
\bibfield{author}{\bibinfo{person}{Jason~D Hartline} {and} \bibinfo{person}{Tim
  Roughgarden}.} \bibinfo{year}{2009}\natexlab{}.
\newblock \showarticletitle{Simple versus optimal mechanisms}. In
  \bibinfo{booktitle}{\emph{Proceedings of the 10th ACM conference on
  Electronic commerce}}. \bibinfo{pages}{225--234}.
\newblock


\bibitem[Hsiao and Chen(2015)]%
        {hsiao2015intent}
\bibfield{author}{\bibinfo{person}{Kuo-Lun Hsiao} {and}
  \bibinfo{person}{Chia-Chen Chen}.} \bibinfo{year}{2015}\natexlab{}.
\newblock \showarticletitle{What drives in-app purchase intention for mobile
  games? An examination of perceived values and loyalty}.
\newblock \bibinfo{journal}{\emph{Electronic Commerce Research and
  Applications}} (\bibinfo{year}{2015}).
\newblock


\bibitem[Huang et~al\mbox{.}(2019)]%
        {huang2019level}
\bibfield{author}{\bibinfo{person}{Yan Huang}, \bibinfo{person}{Stefanus
  Jasin}, {and} \bibinfo{person}{Puneet Manchanda}.}
  \bibinfo{year}{2019}\natexlab{}.
\newblock \showarticletitle{“Level Up”: Leveraging skill and engagement to
  maximize player game-play in online video games}.
\newblock \bibinfo{journal}{\emph{Information Systems Research}}
  \bibinfo{volume}{30}, \bibinfo{number}{3} (\bibinfo{year}{2019}),
  \bibinfo{pages}{927--947}.
\newblock


\bibitem[Intelligence(2023)]%
        {mordor2023}
\bibfield{author}{\bibinfo{person}{Mordor Intelligence}.}
  \bibinfo{year}{2023}\natexlab{}.
\newblock \bibinfo{title}{Global gaming market: 2023 - 2027: Industry share,
  size, growth - mordor intelligence}.
\newblock
\newblock
\urldef\tempurl%
\url{https://mordorintelligence.com/industry-reports/global-gaming-market}
\showURL{%
\tempurl}
\newblock
\shownote{Accessed: 2023-08-06}.


\bibitem[Jiao et~al\mbox{.}(2021)]%
        {jiao2021selling}
\bibfield{author}{\bibinfo{person}{Yifan Jiao}, \bibinfo{person}{Christopher~S
  Tang}, {and} \bibinfo{person}{Jingqi Wang}.} \bibinfo{year}{2021}\natexlab{}.
\newblock \showarticletitle{Selling virtual items in free-to-play games:
  Transparent selling vs. opaque selling}.
\newblock \bibinfo{journal}{\emph{Service Science}} \bibinfo{volume}{13},
  \bibinfo{number}{2} (\bibinfo{year}{2021}), \bibinfo{pages}{53--76}.
\newblock


\bibitem[Joas(2020)]%
        {joas_2020}
\bibfield{author}{\bibinfo{person}{Eino Joas}.}
  \bibinfo{year}{2020}\natexlab{}.
\newblock \bibinfo{title}{Clash of Clans: Bigger, Better, Battle Pass}.
\newblock
\newblock
\urldef\tempurl%
\url{https://www.youtube.com/watch?v=jl9JCmaF9Y0}
\showURL{%
\tempurl}


\bibitem[Keogh and Richardson(2018)]%
        {keogh2018waiting}
\bibfield{author}{\bibinfo{person}{Brendan Keogh} {and} \bibinfo{person}{Ingrid
  Richardson}.} \bibinfo{year}{2018}\natexlab{}.
\newblock \showarticletitle{Waiting to play: The labour of background games}.
\newblock \bibinfo{journal}{\emph{European Journal of Cultural Studies}}
  \bibinfo{volume}{21}, \bibinfo{number}{1} (\bibinfo{year}{2018}),
  \bibinfo{pages}{13--25}.
\newblock


\bibitem[Krasilnikov(2020)]%
        {krasilnikov_2020}
\bibfield{author}{\bibinfo{person}{Vladimir Krasilnikov}.}
  \bibinfo{year}{2020}\natexlab{}.
\newblock \bibinfo{title}{{M}onetization {D}esign: {T}he {D}ark {S}ide of
  {G}acha}.
\newblock
\newblock
\urldef\tempurl%
\url{https://www.youtube.com/watch?v=LnCOkQ-f8AQ&amp;t=1316s}
\showURL{%
\tempurl}


\bibitem[Lehdonvirta(2009)]%
        {lehdonvirta2009virtual}
\bibfield{author}{\bibinfo{person}{Vili Lehdonvirta}.}
  \bibinfo{year}{2009}\natexlab{}.
\newblock \showarticletitle{Virtual Item Sales as a Revenue Model: Identifying
  Attributes that Drive Purchase Decisions}.
\newblock \bibinfo{journal}{\emph{Electronic Commerce Research}}
  \bibinfo{volume}{9}, \bibinfo{number}{1} (\bibinfo{year}{2009}),
  \bibinfo{pages}{97--113}.
\newblock


\bibitem[Levy(2016)]%
        {levy_2016}
\bibfield{author}{\bibinfo{person}{Ethan Levy}.}
  \bibinfo{year}{2016}\natexlab{}.
\newblock \bibinfo{title}{How to design in-game purchases - youtube}.
\newblock
\newblock
\urldef\tempurl%
\url{https://www.youtube.com/watch?v=vc2oNgRPTR0}
\showURL{%
\tempurl}


\bibitem[Levy(2021)]%
        {levy_2021}
\bibfield{author}{\bibinfo{person}{Tammy Levy}.}
  \bibinfo{year}{2021}\natexlab{}.
\newblock \bibinfo{title}{Nature vs. Nurture: Unpacking Player Spending in F2P
  Games}.
\newblock
\newblock
\urldef\tempurl%
\url{https://www.youtube.com/watch?v=HvUn1XOURk4&amp;t=693s}
\showURL{%
\tempurl}


\bibitem[Martin(2008)]%
        {martin2008consuming}
\bibfield{author}{\bibinfo{person}{Jennifer Martin}.}
  \bibinfo{year}{2008}\natexlab{}.
\newblock \showarticletitle{Consuming code: use-value, exchange-value, and the
  role of virtual goods in Second Life}.
\newblock \bibinfo{journal}{\emph{Journal For Virtual Worlds Research}}
  \bibinfo{volume}{1}, \bibinfo{number}{2} (\bibinfo{year}{2008}).
\newblock


\bibitem[Microsoft(2022)]%
        {microsoft_2022}
\bibfield{author}{\bibinfo{person}{Microsoft}.}
  \bibinfo{year}{2022}\natexlab{}.
\newblock \bibinfo{title}{Microsoft to acquire Activision Blizzard to bring the
  joy and community of gaming to everyone, across every device}.
\newblock
\newblock
\urldef\tempurl%
\url{https://news.microsoft.com/2022/01/18/microsoft-to-acquire-activision-blizzard-to-bring-the-joy-and-community-of-gaming-to-everyone-across-every-device/}
\showURL{%
\tempurl}
\newblock
\shownote{Accessed: 2022-02-06}.


\bibitem[Misiuk(2019)]%
        {misiuk_2019}
\bibfield{author}{\bibinfo{person}{Katherine Misiuk}.}
  \bibinfo{year}{2019}\natexlab{}.
\newblock \bibinfo{title}{Free or Paid Mobile Games - What's More Profitable?}
\newblock
\newblock
\urldef\tempurl%
\url{https://indiewatch.net/2019/06/16/free-or-paid-mobile-games-whats-more-profitable/}
\showURL{%
\tempurl}


\bibitem[Moloco(2021)]%
        {moloco2021}
\bibfield{author}{\bibinfo{person}{Moloco}.} \bibinfo{year}{2021}\natexlab{}.
\newblock \bibinfo{title}{How big is the hyper casual games market?}
\newblock
\newblock
\urldef\tempurl%
\url{https://www.moloco.com/en/blog/hyper-casual-games-market-size}
\showURL{%
\tempurl}
\newblock
\shownote{Accessed: 2022-02-06}.


\bibitem[Newzoo(2022)]%
        {newzoo2021}
\bibfield{author}{\bibinfo{person}{Newzoo}.} \bibinfo{year}{2022}\natexlab{}.
\newblock \bibinfo{title}{The games market and beyond in 2021: The Year in
  numbers}.
\newblock
\newblock
\urldef\tempurl%
\url{https://newzoo.com/insights/articles/the-games-market-in-2021-the-year-in-numbers-esports-cloud-gaming/}
\showURL{%
\tempurl}
\newblock
\shownote{Accessed: 2022-02-06}.


\bibitem[Nojima(2007)]%
        {nojima2007MMO}
\bibfield{author}{\bibinfo{person}{Miho Nojima}.}
  \bibinfo{year}{2007}\natexlab{}.
\newblock \showarticletitle{Pricing models and Motivations for MMO play}.
\newblock \bibinfo{journal}{\emph{Digital Games Research Association}}
  (\bibinfo{year}{2007}).
\newblock


\bibitem[Nussbaum and Telfer(2015)]%
        {Nussbaum_Telfer_2015}
\bibfield{author}{\bibinfo{person}{Sebastian Nussbaum} {and}
  \bibinfo{person}{Adam Telfer}.} \bibinfo{year}{2015}\natexlab{}.
\newblock \bibinfo{title}{In it for the Long Hall: How Wooga boosts long-term
  retention}.
\newblock
\newblock
\urldef\tempurl%
\url{https://www.gdcvault.com/play/1022070/In-It-for-the-Long}
\showURL{%
\tempurl}


\bibitem[O'Cass and McEwen(2004)]%
        {cass2004status}
\bibfield{author}{\bibinfo{person}{Aron O'Cass} {and} \bibinfo{person}{Hmily
  McEwen}.} \bibinfo{year}{2004}\natexlab{}.
\newblock \showarticletitle{Exploring consumer status and conspicuous
  consumption}.
\newblock \bibinfo{journal}{\emph{Journal of Consumer Behaviour}}
  \bibinfo{volume}{4}, \bibinfo{number}{1} (\bibinfo{year}{2004}),
  \bibinfo{pages}{25--39}.
\newblock


\bibitem[Partis(2022)]%
        {partis_2022}
\bibfield{author}{\bibinfo{person}{Danielle Partis}.}
  \bibinfo{year}{2022}\natexlab{}.
\newblock \bibinfo{title}{Diablo immortal generates \$100m in lifetime
  revenue}.
\newblock
\newblock
\urldef\tempurl%
\url{https://www.gamesindustry.biz/diablo-immortal-generates-100m-in-lifetime-revenue#:~:text=Diablo%20Immortal%20has%20surpassed%20%24100,two%20months%20after%20its%20release.}
\showURL{%
\tempurl}


\bibitem[Pecorella(2013a)]%
        {gdclongterm2013}
\bibfield{author}{\bibinfo{person}{Anthony Pecorella}.}
  \bibinfo{year}{2013}\natexlab{a}.
\newblock \bibinfo{title}{Building Games for the Long Term: Pragmatic F2P Guild
  Design}.  (\bibinfo{year}{2013}).
\newblock
\urldef\tempurl%
\url{https://www.slideshare.net/Kongregate/building-games-for-the-long-term-pragmatic-f2p-guild-design-gdc-europe-2013}
\showURL{%
\tempurl}
\newblock
\shownote{Accessed: 2022-02-01}.


\bibitem[Pecorella(2013b)]%
        {gdcriseofidle2013}
\bibfield{author}{\bibinfo{person}{Anthony Pecorella}.}
  \bibinfo{year}{2013}\natexlab{b}.
\newblock \bibinfo{title}{The Rise and Rise of Idle Games}.
  (\bibinfo{year}{2013}).
\newblock
\urldef\tempurl%
\url{https://www.slideshare.net/AnthonyPecorella/the-rise-and-rise-of-idle-games-68916528}
\showURL{%
\tempurl}
\newblock
\shownote{Accessed: 2022-01-04}.


\bibitem[Reilley(2023)]%
        {Reilley_2023}
\bibfield{author}{\bibinfo{person}{Don Reilley}.}
  \bibinfo{year}{2023}\natexlab{}.
\newblock \bibinfo{title}{The evolution of Game Monetization, retention and
  acquisition (presented by Amazon prime gaming, Voodoo, Nifty Games, Publicis
  Groupe and Scuti / Lenderwize)}.
\newblock
\newblock
\urldef\tempurl%
\url{https://www.gdcvault.com/play/1029433/The-Evolution-of-Game-Monetization}
\showURL{%
\tempurl}


\bibitem[Sepulveda et~al\mbox{.}(2019)]%
        {sepulveda2019exploring}
\bibfield{author}{\bibinfo{person}{Gabriel~K Sepulveda},
  \bibinfo{person}{Felipe Besoain}, {and} \bibinfo{person}{Nicolas~A Barriga}.}
  \bibinfo{year}{2019}\natexlab{}.
\newblock \showarticletitle{Exploring Dynamic Difficulty Adjustment in Video
  Games}. In \bibinfo{booktitle}{\emph{2019 IEEE CHILEAN Conference on
  Electrical, Electronics Engineering, Information and Communication
  Technologies (CHILECON)}}. IEEE, \bibinfo{pages}{1--6}.
\newblock


\bibitem[Seufert(2017)]%
        {seufert_2017}
\bibfield{author}{\bibinfo{person}{Eric~Benjamin Seufert}.}
  \bibinfo{year}{2017}\natexlab{}.
\newblock \bibinfo{title}{Apps Aren't Dying, but not Everyone Knows How to Make
  Money on Mobile: Mobile Dev Memo by Eric Seufert}.
\newblock
\newblock
\urldef\tempurl%
\url{https://mobiledevmemo.com/apps-arent-dying-making-money-on-mobile/}
\showURL{%
\tempurl}


\bibitem[Sheng et~al\mbox{.}(2020)]%
        {sheng2020incentivized}
\bibfield{author}{\bibinfo{person}{Lifei Sheng},
  \bibinfo{person}{Christopher~Thomas Ryan}, \bibinfo{person}{Mahesh
  Nagarajan}, \bibinfo{person}{Yuan Cheng}, {and} \bibinfo{person}{Chunyang
  Tong}.} \bibinfo{year}{2020}\natexlab{}.
\newblock \showarticletitle{Incentivized actions in freemium games}.
\newblock \bibinfo{journal}{\emph{Manufacturing \& Service Operations
  Management}} (\bibinfo{year}{2020}).
\newblock


\bibitem[Shi et~al\mbox{.}(2015)]%
        {shi2015minnow}
\bibfield{author}{\bibinfo{person}{Savannah~Wei Shi}, \bibinfo{person}{Mu Xia},
  {and} \bibinfo{person}{Yun Huang}.} \bibinfo{year}{2015}\natexlab{}.
\newblock \showarticletitle{From Minnow to Whales: An Empirical Study of
  Purchase Behavior in Freemium Social Games}.
\newblock \bibinfo{journal}{\emph{International Journal of Electronic
  Commerce}}  \bibinfo{volume}{2015} (\bibinfo{year}{2015}),
  \bibinfo{pages}{20--2}.
\newblock


\bibitem[Tafradzhiyski(2023)]%
        {Tafradzhiyski_2023}
\bibfield{author}{\bibinfo{person}{Nayden Tafradzhiyski}.}
  \bibinfo{year}{2023}\natexlab{}.
\newblock \bibinfo{title}{In-app purchases}.
\newblock
\newblock
\urldef\tempurl%
\url{https://www.businessofapps.com/guide/in-app-purchases/#:~:text=In%2Dapp%20purchases%20account%20for,37.8%25%20from%20paid%20app%20downloads.}
\showURL{%
\tempurl}


\bibitem[Thorstein(1899)]%
        {veblen1955}
\bibfield{author}{\bibinfo{person}{Veblen Thorstein}.}
  \bibinfo{year}{1899}\natexlab{}.
\newblock \showarticletitle{The Theory of the Leisure Class}.
\newblock \bibinfo{journal}{\emph{New York: Mentor Books}}
  (\bibinfo{year}{1899}).
\newblock


\bibitem[Turkowski(2023)]%
        {Turkowski_2023}
\bibfield{author}{\bibinfo{person}{Jakub Turkowski}.}
  \bibinfo{year}{2023}\natexlab{}.
\newblock \bibinfo{title}{Types of gamers: Casual and hardcore gamers as target
  groups}.
\newblock
\newblock
\urldef\tempurl%
\url{https://instreamly.com/posts/types-of-gamers-casual-and-hardcore-gamers/}
\showURL{%
\tempurl}


\bibitem[Tversky and Kahneman(1985)]%
        {tversky1985framing}
\bibfield{author}{\bibinfo{person}{Amos Tversky} {and} \bibinfo{person}{Daniel
  Kahneman}.} \bibinfo{year}{1985}\natexlab{}.
\newblock \showarticletitle{The framing of decisions and the psychology of
  choice}.
\newblock In \bibinfo{booktitle}{\emph{Behavioral decision making}}.
  \bibinfo{publisher}{Springer}, \bibinfo{pages}{25--41}.
\newblock


\bibitem[Xue et~al\mbox{.}(2017)]%
        {xue2017dynamic}
\bibfield{author}{\bibinfo{person}{Su Xue}, \bibinfo{person}{Meng Wu},
  \bibinfo{person}{John Kolen}, \bibinfo{person}{Navid Aghdaie}, {and}
  \bibinfo{person}{Kazi~A Zaman}.} \bibinfo{year}{2017}\natexlab{}.
\newblock \showarticletitle{Dynamic difficulty adjustment for maximized
  engagement in digital games}. In \bibinfo{booktitle}{\emph{Proceedings of the
  26th International Conference on World Wide Web Companion}}.
  \bibinfo{pages}{465--471}.
\newblock


\bibitem[Zatkin(2017)]%
        {gdcgamedata}
\bibfield{author}{\bibinfo{person}{Geoffrey Zatkin}.}
  \bibinfo{year}{2017}\natexlab{}.
\newblock \bibinfo{title}{Awesome Video Game Data 2017}.
  (\bibinfo{year}{2017}).
\newblock
\urldef\tempurl%
\url{https://www.gdcvault.com/play/1024054/Awesome-Video-Game-Data}
\showURL{%
\tempurl}
\newblock
\shownote{Accessed: 2022-01-18}.


\end{thebibliography}

\onecolumn
\appendix
\section{Utility Definitions for Fully-Sensitive and Insensitive Players}\label{app:util}

\subsection{Fully-Sensitive Players}
\begin{definition}
At the beginning of a task, the marginal value a fully-sensitive player with type $\type$ has for a skip at price $\skippay$ is
$$
    \valueFunc_m(\typeDist, \skippay) = (1-\type)\valueFunc(\typeDist, \skippay).
$$
The utility a fully-sensitive player receives if they purchase a skip at price $\skippay$ is
$$
    \utilityFunc(\valueFunc(\typeDist, \skippay), \skippay) = \valueFunc(\typeDist, \skippay) - \skippay.
$$
We assume players are utility maximizers. 
Therefore, the utility a fully-sensitive player receives in a task is
$$
    \utilityFunc_{\max}(\valueFunc(\typeDist, \skippay),\skippay,\type) = \max(\type\valueFunc(\typeDist, \skippay), \valueFunc(\typeDist, \skippay) - \skippay).
$$
\end{definition}

\begin{definition}
At the end of each task $\task \in [1, \infty)$, all players realize a retention threshold $\retention_{\task} \sim \retDist$, where $\supp(\retDist) = [0,\infty)$.
A fully-sensitive player quits the game immediately if
$$
    \utilityFuncith{\task}_{\max}(\valueFunc(\typeDist, \skippay_{\task}),\skippay_{\task},\type)<\retention_{\task},
$$
where $\utilityFuncith{\task}_{\max}(\valueFunc(\typeDist, \skippay_{\task}),\skippay_{\task},\type) = \max(\type\valueFunc(\typeDist, \skippay_{\task}), \valueFunc(\typeDist, \skippay_{\task}) - \skippay_{\task})$.
\end{definition}

\subsection{Insensitive Players}
\begin{definition}
At the beginning of a task, the marginal value an insensitive player with type $\type$ has for a skip at price $\skippay$ is
$$
    \valueFunc_m(\skippay) = (1-\type)\valueFunc_c(\skippay).
$$
The utility an insensitive player receives if they purchase a skip at price $\skippay$ is
$$
    \utilityFunc(\valueFunc_c(\skippay), \skippay) = \valueFunc_c(\skippay) - \skippay.
$$
We assume players are utility maximizers. 
Therefore, the utility an insensitive player receives in a task is
$$
    \utilityFunc_{\max}(\valueFunc_c(\skippay),\skippay,\type) = \max(\type\valueFunc_c(\skippay), \valueFunc_c(\skippay) - \skippay).
$$
\end{definition}

\begin{definition}
At the end of each task $\task \in [1, \infty)$, all players realize a retention threshold $\retention_{\task} \sim \retDist$, where $\supp(\retDist) = [0,\infty)$.
An insensitive player quits the game immediately if
$$
    \utilityFuncith{\task}_{\max}(\valueFunc_c(\skippay_{\task}),\skippay_{\task},\type)<\retention_{\task},
$$
where $\utilityFuncith{\task}_{\max}(\valueFunc_c(\skippay_{\task}),\skippay_{\task},\type) = \max(\type\valueFunc_c(\skippay_{\task}), \valueFunc_c(\skippay_{\task}) - \skippay_{\task})$.
\end{definition}

\section{Transformation and Properties of an Insensitive Value Function}
We now give a way to choose a constant value function that well approximates a value function as long as it plateaus enough. 
While the specifics of how this constant value function is chosen are not used in our analysis, we give an example that maintains some nice properties. 
Fix the value of the task at the point $\payment^*$ where $\valueFunc'(\payment^*)=1$ and constrain the designer to set prices $\skippay\geq \skippay^*$.
The constraint is there to avoid the failure case of placing the price on the steep region of the curve, a region which is not well modeled by a flat value function and often gives low utility.
This transformation is done for analytic convenience but it also has a few nice properties: the utility optimal price is guaranteed to be above the truncation point and the utility obtained by buyers under the true value function will be at least as high as the utility on the flat function. 
These facts combine to ensure that such a procedure can only be pessimistic about buyers' utility and will always have retention at least as high when mapped back to the true value functions. 
Furthermore, the revenue of the optimal price on the constant value function is at least a $\nicefrac{\valueFunc_{\mathrm{start}}}{\valueFunc_{\mathrm{end}}}$ fraction of the revenue on the true value function where $\valueFunc_{\mathrm{end}}$ is the value at the end of the low-sensitivity region and  $\valueFunc_{\mathrm{start}}$ the value at the beginning of the low-sensitivity region. 
Therefore, as long as this region is quite flat little is lost by comparing to the optimal point on the constant value function.

\subsection{Transformation Properties}

Here we show why the properties of the transformation listed above are true. 
We introduce some notation that will be useful:
Define $\valueFunc_{\mathrm{end}}$ as the value where $\valueFunc$ takes its maximum value (i.e., where $\valueFunc(\skippay)=\skippay$) and $\valueFunc_{\mathrm{start}}$ as the value where $\valueFunc'(\skippay^*)=1$. 
Define the constant value function $\valueFunc_{c}(\skippay) =\valueFunc_{\mathrm{start}}$.

\subsubsection{The utility optimal point is always above the truncation point} This follows from the fact that the truncation point is at $\valueFunc'(\skippay^*)=1$ and the concavity of $\valueFunc(\skippay)$. 
By concavity, for all $x<\skippay^*$ we must have $\valueFunc'(x)\geq 1$. 
This means utility is increasing for every player whether they purchase or not and so trivially the utility optimal price is higher.

\subsubsection{The utility obtained by buyers under the true value function will be at least as high as the expected utility on the flat function} 
The transformation replaces the true value function with a constant at the lowest price we might charge.
Due to monotonicity of the true value function, we have that for any price we might charge $\valueFunc(\skippay) \geq \valueFunc_c(\skippay)$.
Therefore, $\valueFunc(\skippay) - \skippay \geq \valueFunc_c(\skippay) - \skippay$ for any price $\skippay$.

\subsubsection{The gap in optimal revenue between the constant value function and the true value function.}
First note that the upper bound on the probability of sale is the same for both $\valueFunc(\skippay)$ and $\valueFunc_{c}(\skippay)$; the two functions are equal at the point where $\valueFunc'(\skippay^*)=1$ and probability of sale decreases above this point for both value functions. 
Further, the revenue of the true value function at any price $\skippay$ is upper bounded by $$\min\left(\typeDist\left(1-\frac{\skippay}{\valueFunc_{\mathrm{end}}}\right), ~\typeDist\left(1-\frac{\skippay^*}{\valueFunc_{\mathrm{start}}}\right)\right)\skippay.$$
This is because $$\min\left(\typeDist\left(1-\frac{\skippay}{\valueFunc_{\mathrm{end}}}\right),~\typeDist\left(1-\frac{\skippay^*}{\valueFunc_{\mathrm{start}}}\right)\right)>\typeDist\left(1-\frac{\skippay}{\valueFunc(\skippay)}\right).$$
Therefore, for any price $\skippay$ on the original value function we can set a price of $$\max\left(\skippay^*,~\frac{\valueFunc_{\mathrm{start}}}{\valueFunc_{\mathrm{end}}}\skippay\right)$$ on the constant value function and achieve at least a $\nicefrac{\valueFunc_{\mathrm{start}}}{\valueFunc_{\mathrm{end}}}$ fraction of the revenue per round while providing players with the same utility and therefore, the same retention. Finally, since there exists a price on the constant value function that provides a $\nicefrac{\valueFunc_{\mathrm{start}}}{\valueFunc_{\mathrm{end}}}$ approximation to the optimal revenue on the true value function, the optimal price on the constant value function must provide at least a $\nicefrac{\valueFunc_{\mathrm{start}}}{\valueFunc_{\mathrm{end}}}$ approximation to the optimal revenue on the true value function.

\subsubsection{Dropping the constraint does not impact analysis}\label{appendix:cons_value}
When face with a minimum payment constraint the Myerson optimal payment on an MHR distribution is simply the max of the constraint and the Myerson optimal payment. Similarly the known types designer cannot sell to any of the players whos marginal utility is below $v-\skippay^*$ and must choose a retention threshold $\skipret\geq\skippay^*$. It is easy to see the main theorem follows by the same steps as before by simply comparing the relative performance of known types and MT pricing on the remaining portion of the population which could ever be sold to.  

\section{Complex Pricing Schemes}\label{appendix:complex_pricing}
We now show what happens if a designer was able to set multiple prices for players based on how much of the task they complete, for example charging less if the player waits for half the task.
We can show that in a given round such a pricing scheme can generate more revenue.

\begin{theorem}\label{thm:multivssingle}
There exists type distributions $\typeDist$ such that the revenue achieved by selling multiple skips is strictly greater than the optimal single-skip revenue.
\end{theorem}
\begin{proof}
    We prove this with a construction. 
    Imagine a task is segmented into two blocks, the marginal value for skipping just the last block is $(1-\type)\valueFunc$ and the value for skipping both blocks is $(1-\type^2)\valueFunc$. 
    Let the value for the task be $\valueFunc=1$, $\skippayith{0}$ be the price for skipping both blocks and $\skippayith{1}$ the price for skipping only the last block.
    We define a discrete type distribution with $\Prx{\type = 0.25}=0.5$ and $\Prx{\type = 0.75}=0.5$. 
    The pricing problem becomes a bit trickier with multiple blocks; we need to know how much \emph{projected} utility a player would receive for buying a skip at each block to avoid cannibalization. 
    If a player buys a skip immediately, they would receive utility $(1-\type^2)-\skippayith{0}$, and for waiting and then buying they would get $\type((1-\type) -\skippayith{1})$.
    We set prices $\skippayith{0} = \frac{13}{16}$ and $\skippayith{1} = \frac{1}{4}$. 
    The low type buys a skip in the second round, and the high type prefers buying in the first round: 
    \begin{center}
      \begin{tabular}{c|c c}
        \centering
         & $(1-\type^2)-\skippayith{0}$ &
         $\type((1-\type) -\skippayith{1})$ \\
         \hline
         $\type_{h} = 0.25$ & 0.125 & 0.125\\
         $\type_{\ell} = 0.75$ & -0.375 & 0 
      \end{tabular}  
    \end{center}
    Therefore, we achieve revenue $\REV = 0.5(0.8125)+0.5(0.25) \approx 0.53$. 
    On the other hand, single skip may choose to offer a price at either of the two types' marginal values in the first round. 
    This means they can choose either $\skippayith{0} = 0.9375$ or a price of $\skippayith{0} = 0.4375$ which achieve revenue $\REV = 0.46875$ and $\REV = 0.4375$ respectively.
\end{proof}

\section{Proofs from Section 5} \label{app:proof_sec_5}
\begin{lemma}\label{lma:truncMHR}
    If a distribution $\dist(x)$ satisfies MHR, then the distribution truncated to any interval $[a,b]$ also satisfies MHR.
\end{lemma}
\begin{proof}[Proof of \Cref{lma:truncMHR}]
The resulting distribution after truncation is $\frac{\dist(x)-\dist(a)}{\dist(b)-\dist(a)}$ and its corresponding density function is $\frac{\dense(x)}{\dist(b)-\dist(a)}$. We can write the inverse hazard rate as 
\begin{align*}
\left(1-\frac{\dist(x)-\dist(a)}{\dist(b)-\dist(a)}\right)\frac{\dist(b)-\dist(a)}{\dense(x)} =
\frac{\dist(b)-\dist(x)}{\dense(x)}
\end{align*}
We need to ensure the derivative is always non-positive,
\begin{align*}
\frac{\diff}{\diff x}\frac{\dist(b)-\dist(x)}{\dense(x)} = -\dense(x)^{2} - (\dist(b)-\dist(x))(\dense'(x))
\end{align*}
We must show
\begin{align*}
-\dense(x)^{2} - (\dist(b)-\dist(x)(\dense'(x)) \leq& 0 \\
-\frac{\dense'(x)}{\dense(x)^{2}} \leq& \frac{1}{\dist(b)-\dist(x)} \\
\leq& \frac{1}{1-\dist(x)}
\end{align*}
Its easy to see that this resulting condition is guaranteed by simply taking the derivative of the original hazard rate which satisfies MHR.
\end{proof}

\begin{lemma}\label{lma:scalarMHR}
    If a distribution $\dist(x)$ satisfies MHR, then a distribution whose input $x$ is scaled by some constant $c>0$ also satisfies MHR.
\end{lemma} 
\begin{proof}[Proof of \Cref{lma:scalarMHR}]
    We have that directional monotonicity is preserved under composition i.e., for $g(x)$ monotonically increasing, we have that $f(x)$ monotonically increasing implies $f(g(x))$ is monotonically increasing.
    Let the hazard rate of the original distribution be $f(x)$ and the scaled input to that distribution be $g(x) = cx$.
    Clearly, $g(x)$ is monotonically increasing for all $c>0$ and by definition of MHR, $f(x)$ is monotonically increasing.
\end{proof}

\begin{lemma}\label{lma:MHR}
    If $\dist_{1-\type}$ has MHR then the marginal value distribution $\margDist$ has MHR.
\end{lemma}

\begin{proof}
    The hazard rate of the $1-\type$ distribution, $H_{1-\type}$, can be derived directly from the hazard rate of the marginal value distribution, $H_{m}$.
    \begin{align*}
    H_m(x) &= \frac{\margDense\left(x\right)}{1-\prob[m]{(1-\type)\valueFunc \leq x}} = \frac{\dense_{1-\type}\left(\frac{x}{\valueFunc}\right)}{1-\prob[1-\type]{1-\type \leq \frac{x}{\valueFunc}}} \cdot \left|\frac{\diff}{\diff x} \frac{x}{v}\right| = \frac{\dense_{1-\type}\left(\frac{x}{\valueFunc}\right)}{1-\dist_{1-\type}\left(\frac{x}{\valueFunc}\right)} = H_{1-\type}
    \left(\frac{x}{\valueFunc}\right).
    \end{align*}
    As we can see, $H_{1-\type}(\nicefrac{x}{\valueFunc})$ is simply a scalar transformation on the input to $H_m(x)$ and from Proposition \ref{lma:scalarMHR} we have that MHR is preserved under such a transformation. 
\end{proof}

\section{Simulation Parameters}\label{appendix:simparams}
For all settings, the value of completing the tasks was set to $1$. 
We also choose $\skippay^*=0$. 
We use an initial population size of $N=10,000,000$.
For player type distributions we used exponential distributions with $\lambda=1,2,3$ as well as a uniform distribution on $[0,1]$. 
For retention distributions we used exponential with $\lambda=1,3,5$ and Pareto distributions with $\alpha = 3,5$. 
For game designer discount factors we used $\beta = 0.97,0.99,0.999$. 
For virality we used growth rates of $1\%$ and $5\%$. 
We first tested tested all possible combinations of these configurations both with and without independent retention draws, however we dropped the value of $\lambda=2$ for the player distribution when used in conjunction with the Pareto retention distribution to increase the space of other parameters we could explore. 
We chose these parameters both because they reflect the whale distributions from \Cref{sec:util} as well as to give us good coverage over the type of retention thresholds we might get.

For the scaling experiments we used the combinations of exponential retention distributions with the same $
\lambda$ parameters as above with the uniform player type distribution.
We used these with every combination of $\beta$ and also with each possible setting of virality. 

\subsubsection{Population Growth}
We also examine the impact of new players joining the game on MT pricing's performance with a ``viral'' model of population growth. 
In this model, there is a constant probability that each player will recruit a friend in each round. 
Each recruited player is drawn from the same type distribution i.i.d. as the original cohort of players. 

\subsubsection{Independently Drawn Retention}

We now test MT pricing when players sample their own retention threshold independently;
each player $i$ draws $\retention_{\task,i}$ i.i.d. from the retention distribution $\retDist$.
A player stays in the game if $\utilityFunc_{\max}(\valueFunc,\skippay_{\task},\type_{i})\geq \retention_{\task,i}$.
A fruitful elaboration to our model of retention would be somewhere between fully correlated and independent retention draws. 
This would model both players unique preferences but also the presence of popular competitors in the same niche genres.
See \Cref{fig:indepretHistogram} for a histogram of performance ratios.

\subsection{Lowering Payments}
We try three scaling factors on the Myerson price: $c\in\{\nicefrac{2}{3}, \nicefrac{1}{2}, \nicefrac{1}{3}\}$. 
We see that $c=\nicefrac{2}{3}$ has some revenue gains and is relatively safe; it does not decrease performance in any of the settings we tested.
Scaling the Myerson price by $\nicefrac{1}{2}$ had larger variance in performance, performing well in settings with the lowest threshold prices while in others it performed worse than no scaling.
Finally, scaling by $\nicefrac{1}{3}$ was worse across all the settings we tested.
Full histograms of the relative performance of scaled to unscaled MT pricing can be found in the appendix (\Cref{fig:scaledMThistogram}).

\section{Supplementary Images} \label{appendix:images}
\begin{figure}[ht]
    \centering
    \includegraphics[trim={0 0cm 1cm 0}, clip, width=0.3\textwidth]{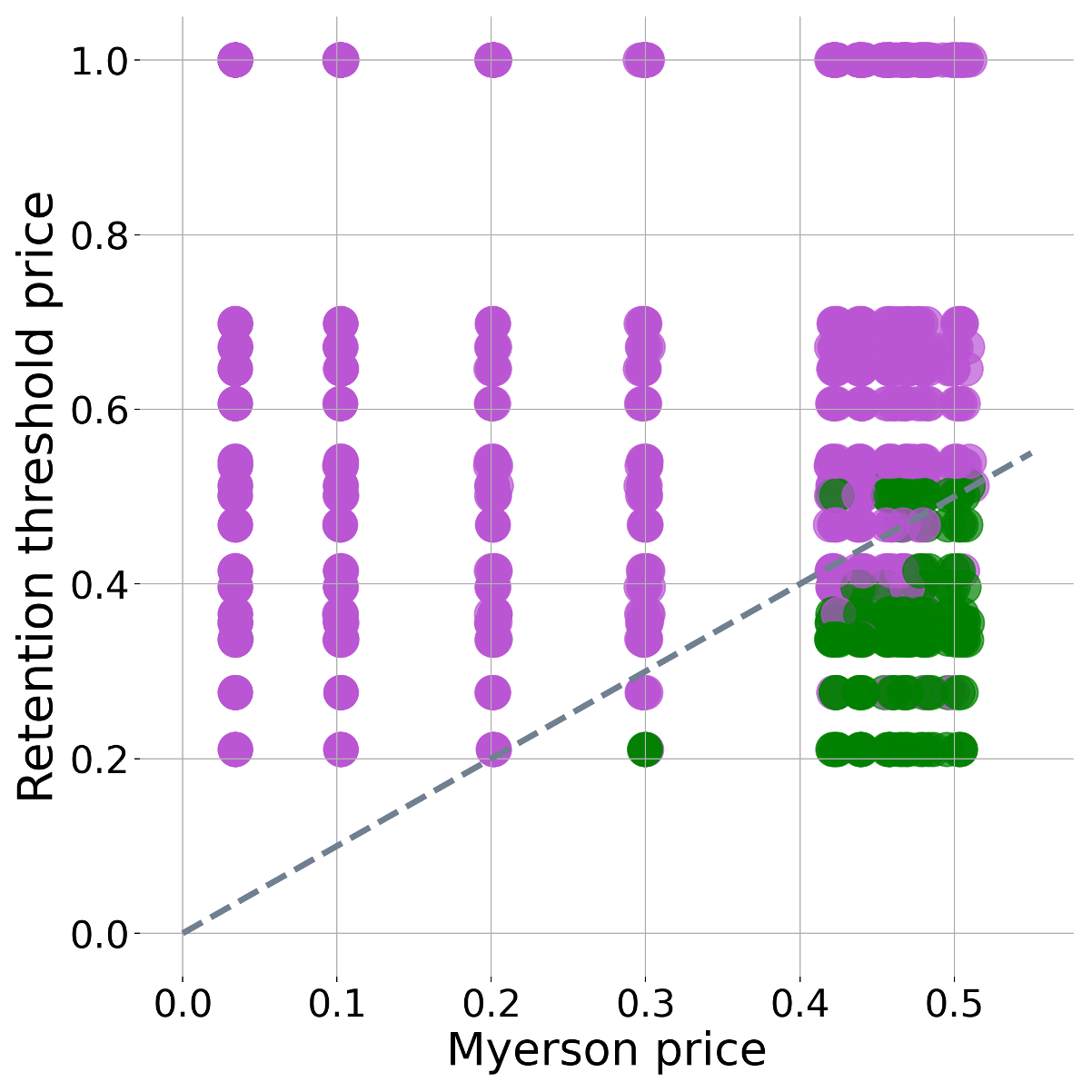}
    \caption{First round prices of Myerson optimal and retention threshold revenues across a variety of simulation parameters.}
    \label{fig:alt_pricing_revs}
\end{figure}

\begin{figure}[t]
    \centering
    \includegraphics[width=\textwidth]{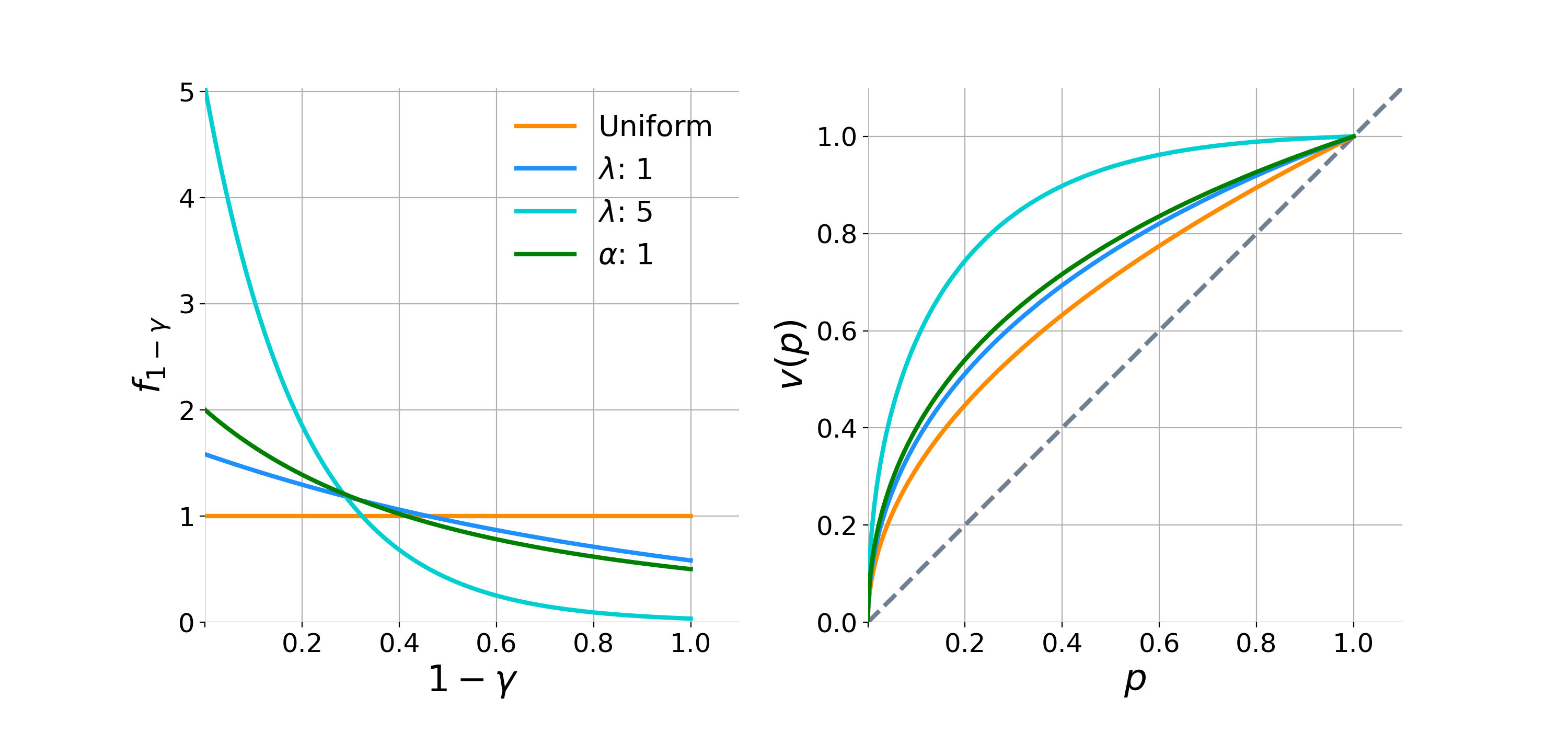}
  \caption{Example value functions generated from their corresponding distributions } \label{fig:example_value_dists}
\end{figure}

\begin{figure}[t]
    \centering
    \includegraphics[width=\textwidth]{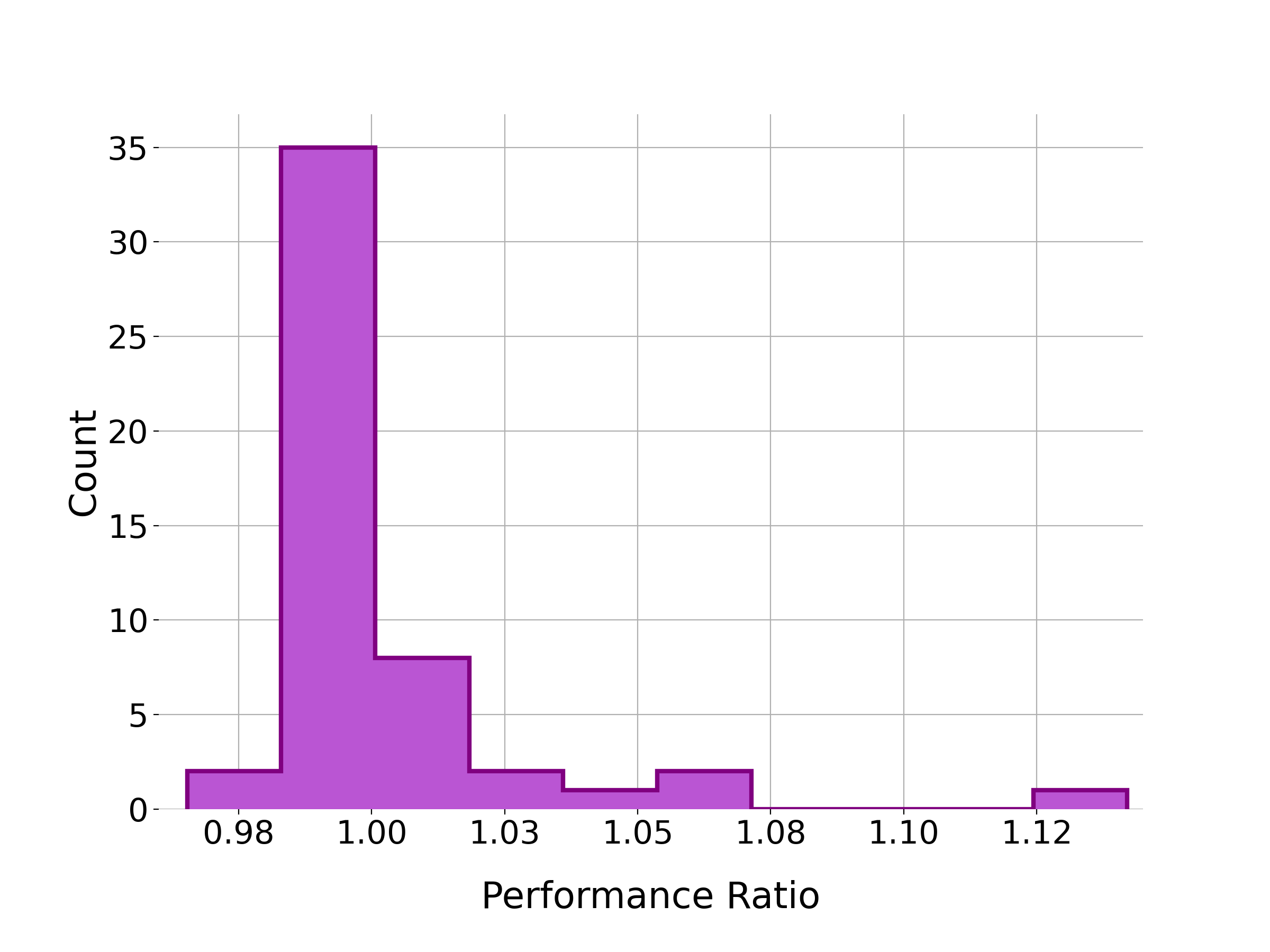}
  \caption{Histogram of the Performance Ratio between MT and Best Alternative (i.e. $\nicefrac{\REV(MT)}{\max(\REV(\text{Myerson Pricing}),\REV(\text{Retention Threshold}))}$)} \label{fig:MThistogram}
\end{figure}
\begin{figure}[t]
    \centering
    \includegraphics[width=\textwidth]{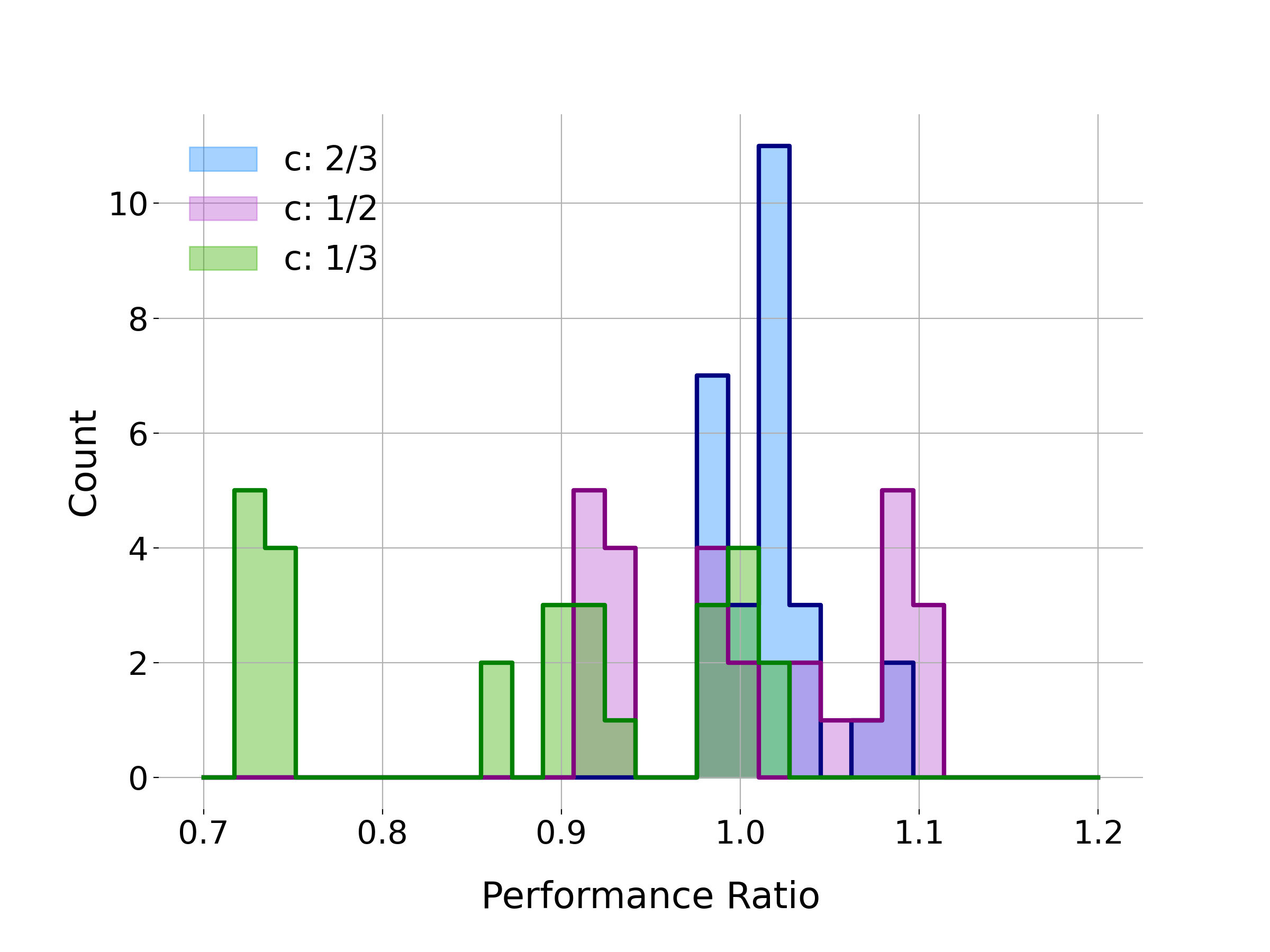}
  \caption{Histogram of the Performance Ratio between MT and Scaled MT (i.e. $\frac{\REV(Scaled MT)}{\REV(MT)}$) with Scaling Factors $c\in\{\nicefrac{2}{3},\nicefrac{1}{2},\nicefrac{1}{3}\}$} \label{fig:scaledMThistogram}
\end{figure}
\begin{figure}[t]
    \centering
    \includegraphics[width=\textwidth]{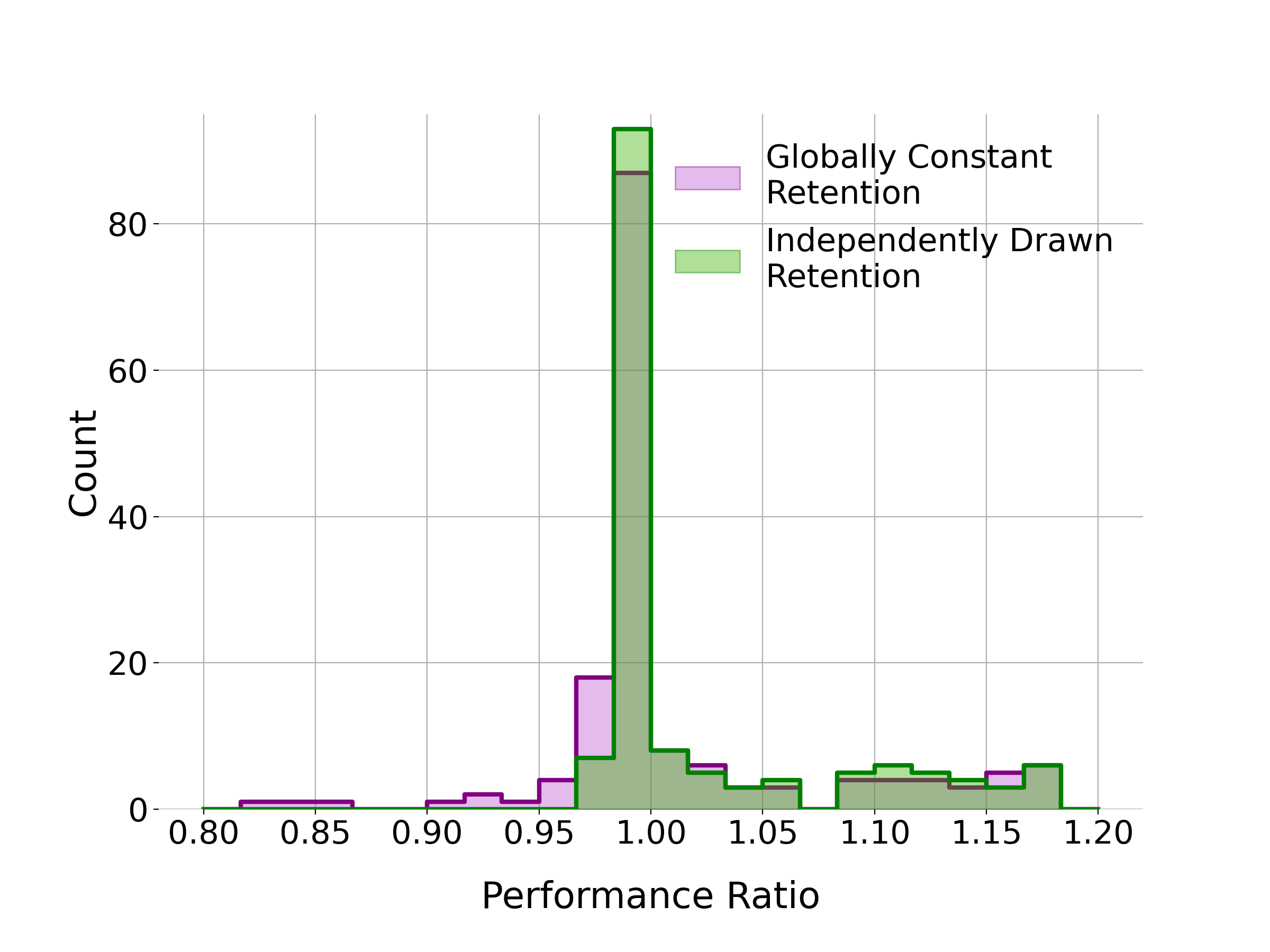}
  \caption{Histogram of the Performance Ratio between MT Pricing and Best Alternative (i.e. $\nicefrac{\REV(MT)}{\max(\REV(\text{Myerson Pricing}),\REV(\text{Retention Threshold}))}$) }
  \label{fig:indepretHistogram}
\end{figure}
\end{document}